\documentclass[12pt]{article}
\usepackage{geometry}
\usepackage{longtable}
\usepackage{booktabs}
\usepackage{multirow}
\usepackage{setspace}

\usepackage[authoryear,round,colon,sort&compress]{natbib}
\bibpunct{(}{)}{,}{autheryear}{,}{,}

\usepackage[english]{babel}
\usepackage{amsmath, amsthm, amssymb, amscd, amsfonts, bm, bbm, mathrsfs}   
\usepackage{graphicx}
\usepackage{enumerate}
\usepackage[shortlabels]{enumitem}
\usepackage{url} 
\usepackage[colorlinks]{hyperref}
\usepackage{xr-hyper}
\usepackage[capitalize,nameinlink,noabbrev]{cleveref} 
\usepackage[algo2e,linesnumbered,lined,ruled]{algorithm2e}

\usepackage[bottom]{footmisc}
\usepackage[normalem]{ulem}
\usepackage[dvipsnames]{xcolor}
\hypersetup{
	breaklinks,
	linkcolor=blue,
	urlcolor=blue,
	anchorcolor=blue,
	citecolor=blue
}

\usepackage{rotating}

\makeatletter

\theoremstyle{plain}

\@ifundefined{date}{}{\date{}}
\usepackage{babel}

\allowdisplaybreaks

\usepackage{tikz}
\usetikzlibrary{fit,positioning}

\allowdisplaybreaks

\newcommand*{\indep}{%
  \mathbin{%
    \mathpalette{\@indep}{}%
  }%
}
\newcommand*{\nindep}{%
  \mathbin{
    \mathpalette{\@indep}{\not}
  }%
}
\newcommand*{\@indep}[2]{%
  \sbox0{$#1\perp\m@th$}
  \sbox2{$#1=$}
  \sbox4{$#1\vcenter{}$}
  \rlap{\copy0}
  \dimen@=\dimexpr\ht2-\ht4-.2pt\relax
  \kern\dimen@
  {#2}%
  \kern\dimen@
  \copy0 
}

\newtheorem{thm}{Theorem}

\theoremstyle{definition}
\newtheorem{rem}{Remark}

\newtheorem{defn}{Definition}

\makeatletter
\newcommand*{\rom}[1]{\expandafter\@slowromancap\romannumeral #1@}
\newcommand{\HUGE}{\@setfontsize\Huge{40}{50}}   
\makeatother

\makeatletter
\newcommand{\labelcustom}[2]{%
	\protected@write \@auxout {}{\string \newlabel {#1}{{#2}{\thepage}{#2}{#1}{}} }%
	\hypertarget{#1}{#2}
}
\makeatother

\makeatletter
\newcommand{\labeltext}[3][]{%
	\@bsphack%
	\csname phantomsection\endcsname
	\def\tst{#1}%
	\def\refmarkup{}%
	\ifx\tst\empty\def\@currentlabel{\refmarkup{#2}}{\label{#3}}%
	\else\def\@currentlabel{\refmarkup{#1}}{\label{#3}}\fi%
	\@esphack%
	\labelmarkup{#2}
}
\makeatother

\makeatletter
\newcommand{\bianca}{\renewcommand\NAT@open{[}\renewcommand\NAT@close{]}}
\makeatother

\newcommand{\pr}{\mathsf{P}}
\newcommand{\eo}{\mathsf{E}}
\newcommand{\cov}{\mathsf{cov}}
\newcommand{\var}{\mathsf{var}}

\newcommand{\nd}{\mathsf{N}}

\newcommand{\ap}{\alpha} 
\newcommand{\g}{\gamma} 
\newcommand{\ga}{\Gamma} 
\newcommand{\dt}{\delta} 
\newcommand{\e}{\varepsilon} 
\newcommand{\Oa}{\Omega} 
\newcommand{\s}{\sigma} 
\newcommand{\ld}{\lambda} 
\newcommand{\Ld}{\Lambda} 

\newcommand{\HH}{\mathbb{H}} 
\newcommand{\I}{\mathbb{I}} 
\newcommand{\R}{\mathbb{R}} 
\newcommand{\SSS}{\mathbb{S}} 
\newcommand{\W}{\mathbb{W}} 
\newcommand{\X}{\mathbb{X}} 

\newcommand{\aaa}{\mathcal{A}}	
\newcommand{\cc}{\mathcal{C}}	
\newcommand{\hh}{\mathcal{H}}	
\newcommand{\ii}{\mathcal{I}}	
\newcommand{\nn}{\mathcal{N}}	
\newcommand{\rr}{\mathcal{R}}	
\newcommand{\ttt}{\mathcal{T}}	
\newcommand{\uu}{\mathcal{U}}	
\newcommand{\vv}{\mathcal{V}}	
\newcommand{\xx}{\mathcal{X}}	
\newcommand{\yy}{\mathcal{Y}}	

\let\oldnl\nl
\newcommand{\nlnonumber}{\renewcommand{\nl}{\let\nl\oldnl}}
\newcommand{\KwTune}{\KwSty{Parameters:}}

\renewcommand{\eqref}[1]{(\ref{#1})}

\begin{document}

	\renewcommand{\sectionautorefname}{Section}
	\renewcommand{\subsectionautorefname}{Section}
	\renewcommand{\subsubsectionautorefname}{Section}
	\renewcommand{\algorithmautorefname}{Algorithm}

\title{
	Conditional regularized halfspace depth \\
	for sparse functional data 
	and its applications
}
\author{
	Hyemin Yeon\thanks{
			Department of Mathematical Sciences, Kent State University, Kent, OH, 44242, USA, Email: hyeon1@kent.edu.
		}, 
            Xiongtao Dai\thanks{
			Division of Biostatistics, University of California, Berkeley, CA, 94720, USA, Email: xiongtao.dai@hotmail.com.
		}, and
		Sara Lopez-Pintado\thanks{
			Department of Health Sciences, Northeastern University, Boston, MA, 02115, USA, Email: s.lopez-pintado@northeastern.edu.
		}
}

\maketitle
\begin{abstract}
	\noindent
	Many functional datasets are observed sparsely and irregularly.
	Ordering such data is challenging because only limited information is available from each observation, while the underlying trajectories remain infinite-dimensional. 
	This paper develops a novel depth notion for sparse functional data, called the \emph{conditional regularized halfspace depth} (CRHD). 
	CRHD is defined as the infimum of conditional halfspace probabilities of the underlying trajectory given the observed sparse measurements, 
	thereby enabling depth evaluation directly at sparse observations without requiring trajectory reconstruction. 
	We study several basic theoretical properties of CRHD that clarify its behavior as a depth measure. 
	The proposed depth is applicable even to extremely sparsely observed functional data, overcoming key limitations of existing sparse functional depths that often rely on reconstructed curves.
	In addition, CRHD induces meaningful rankings for complex functional data. 
	Its numerical performance is demonstrated through rank-based tests, and its practical utility is illustrated using an infant growth dataset.
	
\vspace{0.1in}

\end{abstract}



\clearpage
\section{Introduction}


\subsection{Data depth for dense functional data}

Functional data have become increasingly prevalent in contemporary statistics, with a rich history spanning more than three decades.
Each data unit is a function, such as a curve or trajectory, defined over a compact interval, and the main complication arises from the infinite dimensionality of the underlying function space.
Numerous methodologies have been proposed to analyze functional data, beginning with the pioneering works of \cite{BR86, RD91, RS91, Silver95, Silver96}; functional principal component analysis (FPCA), a central tool in this field, has even earlier roots in \cite{Kleffe73, CLS86}.
Several textbooks summarize the broad applicability of functional data approaches, ranging from classical methods \citep{SR02, RS05} to more recent advances \citep{HK12, KR17, CGLC24}, while also providing comprehensive overviews of key theoretical results in functional data analysis (FDA) \citep{HE15}.

Constructing a reasonable and rigorous ordering structure for functional data is an important problem, particularly for extending rank-based statistical methods to FDA.
A notion of data depth assigns a non-negative real value to each observation in the function space, measuring how central or representative the observation is relative to a distribution or sample.
Larger depth values correspond to more central observations, and functional observations can consequently be ranked from the center outward according to their depths.
When functional data are fully observed, hereafter referred to as \emph{dense functional data}, the first functional depth was introduced by \cite{FM01}.
Since then, many depth notions for dense functional data have been developed, including those of \cite{CN08, LR09, LR11, LSLG14, CC14b, CHSV14, NN16, NGH17, Nagy21}.
Some of these depths can be viewed as non-trivial extensions of finite-dimensional data depth notions in Euclidean spaces, such as the halfspace depth of \cite{tukey75}.

Once a depth-based ranking structure for functional data is available, a wide range of statistical applications becomes possible.
These include location estimation through medians or trimmed means \citep{FM01, CHSV14}, rank-based hypothesis testing \citep{LR09,CC15,GKN24,RC24}, classification \citep{Nagy21, HGON15}, and outlier detection \citep{SG11, HRS15, NN16, NGH17}, among many others.

\subsection{Data depth for sparse functional data}

When functional data are sparsely and irregularly observed and contaminated by noise, constructing a meaningful ranking becomes substantially more challenging.
The difficulty arises not only from the infinite-dimensional nature of the underlying trajectories, but also from the limited information available at sparse observation times and the additional uncertainty introduced by measurement error.
Such data, commonly referred to as \emph{sparse functional data}, frequently arise in longitudinal studies; see, for example, \cite{YMW05a, YMW05b, Muller05, DMT18, GDM21}.
Although many statistical problems for sparse functional data have been extensively studied, including FPCA \citep{YMW05a}, regression \citep{YMW05b}, mean and covariance estimation \citep{LH10, ZW16}, derivative estimation \citep{DMT18}, classification \citep{Muller05, WL13}, and clustering \citep{JS03}, data depth for sparse functional data remains relatively underexplored in the functional depth literature.

A common approach to assigning depth values to sparse functional data is based on a two-stage procedure.
One type of two-stage approach first extracts functional principal component (FPC) scores from sparsely observed data and then applies classical finite-dimensional depth methods, such as halfspace, simplicial, or projection depths \citep{ZS00}, to a finite number of estimated FPC scores.
Another type first reconstructs or smooths the sparse observations and then evaluates existing dense functional depths, such as the modified band depth \citep{LR09}, integrated depth \citep{CHSV14, NGH17}, or extremal depth \citep{NN16}, at the predicted smooth curves.
The latter strategy has been discussed, in particular, by \cite{LW11, LQ21}.
Reconstruction of sparse functional data typically relies on FPCA-based methods, such as Principal Analysis by Conditional Estimation (PACE) \citep{YMW05a, YMW05b, Muller05, DMT18, GDM21, fdapace}, or on FPCA under a mixed-model framework \citep{GGC13,refund}; see also \cite{CY11, Kraus15, KL20, LW22, ZM25} for other reconstruction approaches.

Although such two-stage methods can provide useful ordering structures, they do not directly incorporate the uncertainty induced by sparse and noisy observations when computing depth values.
Consequently, they may be less effective for depth-based statistical analyses of sparse functional data, because the final depth calculation relies only on predicted curves or estimated FPC scores.
This limitation is particularly relevant when the purpose of the analysis is to recover a reliable ranking of the underlying trajectories, since reconstruction error or score-estimation error can distort the final ordering.
Similar limitations of two-stage approaches based on reconstructed data have been discussed in other functional data contexts, including derivative estimation \citep{LM09, DMT18}, hypothesis testing \citep{ZYW24, ZW23}, and classification \citep{DH12}.

To our knowledge, the first attempt to address this issue in the sparse functional depth literature is the modified band depth under uncertainty, denoted by $\text{MBD}_U$, proposed by \cite{SL21}.
This method averages the modified band depth \citep{LR09} over three functions: the reconstructed curve and the upper and lower bounds of its confidence band, constructed using the FPCA method of \cite{GGC13}.
This idea was further extended to multivariate functional data by \cite{QG22}.
Although these methods incorporate uncertainty through confidence bands, they still rely heavily on a two-stage procedure and therefore may not fully exploit the information contained in the sparse observation scheme itself.
Another recent proposal by \cite{Elias23} introduces the partially observed integrated functional depth (POIFD), a modification of the classical integrated depth designed to avoid reconstruction when computing depth; this depth was further extended in \cite{QDG22}.
While POIFD represents an important advancement, it may still inherit limitations from its dense-data counterpart when applied to complex functional data.
These considerations motivate the development of a depth notion that can be evaluated directly at sparse functional observations, while retaining the geometric appeal of halfspace-type depths and addressing the infinite-dimensional nature of functional data.

\subsection{Proposed approach and contributions}

The main contribution of this paper is methodological.
We propose a novel depth measure for sparse functional observations, called the \emph{conditional regularized halfspace depth} (CRHD), designed to address the challenges posed by both infinite dimensionality and sparse observation schemes.
Specifically, CRHD is defined as the infimum of the \emph{conditional} halfspace probabilities of the underlying true trajectory given the sparse observation at which the depth is evaluated.
This construction enables depth evaluation directly at sparse functional observations, without requiring reconstruction of the full trajectory.
The proposed depth is formulated within a Hilbert-space framework for functional data, and its central idea extends the conditional-expectation principle underlying the well-known PACE methodology to the problem of depth evaluation.
To the best of our knowledge, these features have not been simultaneously achieved by existing approaches to depth for sparse functional data.

The proposed CRHD differs fundamentally from the regularized halfspace depth (RHD) of \cite{YDL25RHD}, which was developed for densely observed functional data.
First, the original RHD cannot be directly applied to sparse functional observations; it requires prior reconstruction of the underlying trajectories, which can reduce the accuracy of subsequent statistical analyses.
Second, from a methodological perspective, CRHD replaces the empirical counting of observations within halfspaces used in RHD with the computation of conditional distributions based on estimated model parameters, making its implementation substantially more challenging (cf.~\autoref{ssec_2_3}).
This distinction is particularly evident in our numerical studies, where applying RHD to predicted curves fails to preserve the RHD-induced rankings of the true underlying functions.
Finally, the theoretical analysis of CRHD requires a more delicate treatment, as its definition is based on conditional probabilities rather than empirical halfspace probabilities (cf.~\autoref{ssec_2_2}).

Several steps are required to compute CRHD in practice and to address the combined challenges of sparse observation schemes and infinite dimensionality.
First, to make the estimation of the relevant conditional distributions tractable, we adopt a Gaussian working model.
Importantly, this assumption is used only to estimate the conditional distributions, not to define CRHD itself, and our numerical studies indicate that the resulting procedure remains stable under non-Gaussian data-generating mechanisms (cf.~\autoref{sec4}).
Second, we estimate the model parameters---the mean function, covariance function, and noise variance---using PACE \citep{YMW05a}, although other estimation methods could also be incorporated \citep[e.g.,][]{refund, Kraus15, KL20, LW22, ZM25}.
Finally, because CRHD is defined through an infimum over projection directions in an infinite-dimensional space, we approximate this infimum using a random projection approach.
Further discussion of the practical computation is provided in \autoref{secConclusion}.


We demonstrate the practical value of CRHD through rank-based tests.
Our main empirical comparison focuses on the proposed CRHD and existing two-stage approaches, since the other depth notions discussed above are not readily implementable in the sparse functional setting for various reasons; see \autoref{remCompare} in \autoref{sec4} for details.
The simulation studies show that direct depth evaluation via CRHD consistently outperforms, or remains competitive with, two-stage approaches based on trajectory reconstruction.
CRHD is particularly effective in settings where the analysis is sensitive to extremeness.
This advantage arises from its ability to account for sparse observation schemes directly, including for extreme observations, which two-stage approaches may fail to recover accurately.

We further apply the CRHD-based rank test to the infant growth dataset, including a modified version in which the two groups of growth curves are made to overlap more substantially than in the original data.
This modification emphasizes differences in curve shape rather than overall magnitude.
The CRHD-based rank test proves useful in such challenging sparse functional data settings, where existing methods may have limited power.
For example, our method detects group differences in both the original and modified datasets, whereas the method used in the existing analysis of \cite{LQ21} detects the difference only in the original dataset.
This data example illustrates the practical utility of CRHD for both relatively well-separated and more complex sparse functional data.

\subsection{Outline}

The rest of the paper is outlined as follows.
In \autoref{sec2}, after a brief overview of the original RHD for dense functional data, 
we formally define the proposed CRHD, provide its theoretical properties, and introduce its sample counterparts. 
\autoref{sec3} then devotes a detailed description of its practical computation. 
In \autoref{sec4}, we provide comprehensive numerical studies,
with a main emphasis on contrasting our approach with the two-stage methods. 
The proposed CRHD is illustrated by using a growth dataset in \autoref{sec5}.
The proposed methods are implemented by R function \texttt{CRHD}, 
which is appended to the R package \texttt{RHD} available via the author's Github repository: \texttt{https://github.com/luckyhm1928/RHD}.
Finally, the supplementary document includes further theoretical and numerical details.


\section{Conditional regularized halfspace depth}
\label{sec2}

After \autoref{ssec_2_1} reviews the RHD by \cite{YDL25RHD},
we define the proposed depth, which extend the regularized halfspace depth to sparse functional data, and explore its theoretical properties in \autoref{ssec_2_2}.
\autoref{ssec_2_3} then describe its estimation procedure.

\subsection{Regularized halfspace depth for dense functional data} \label{ssec_2_1}


The classical halfspace depth \citep{tukey75}, when directly applied to infinite-dimensional data, may encounter the degeneracy \citep{DGC11}:
namely, depth values are zero almost everywhere. 
The RHD overcomes this degeneracy issue
by restricting the projection direction set in the halfspace depth to directions with RKHS norms less than a pre-specified positive number $\ld \in (0,\infty)$. 
In addition, the RHD is shown to be practically beneficial as it effectively identifies shape features of functional data.
In this section, we will provide a brief overview of the regularized halfspace depth 
when functional data are fully observed without contamination. 

Let $\HH$ be an infinite-dimensional separable Hilbert space equipped with inner product $\langle \cdot, \cdot \rangle$ and the induced norm $\|\cdot\| = \sqrt{\langle \cdot, \cdot \rangle}$, and $X$ be a random function that takes values in $\HH$
with mean element $\mu \equiv \eo[X] \in \HH$.
We suppose the finite second moment of $X$ as $\eo[\|X\|^2]<\infty$,
and define the covariance operator of $X$ by $\ga = \var[X] \equiv \eo[(X - \mu) \otimes (X-\mu)]$.
Here, we define the bounded linear operator $x \otimes y:\HH \to \HH$, which is called the tensor product of two elements $x,y \in \HH$, by $(x \otimes y)(z) = \langle z, x \rangle y$ for each $z \in \HH$.
Since $\ga$ is self-adjoint, non-negative definite, and compact, 
by spectral decomposition in general Hilbert space,
the covariance operator admits the eigen-decomposition as $\ga = \sum_{k=1}^\infty \g_k (\phi_k \otimes \phi_k)$ \citep{mac09, HE15}.
Here, $\g_k$ and $\phi_k$ respectively denote the $k$-th eigenvalue and eigenfunction of $\ga$,
where the eigenvalues $\{\g_k\}_{k=1}^\infty$ satisfy $\g_1 \geq \g_2 \geq \cdots \geq 0$ with $\g_k \to 0$ as $k\to\infty$ and the eigenfunctions $\{\phi_k\}_{k=1}^\infty$ form a complete orthonormal system for the closure of the image of $\ga$. 
We proceed without loss of generality under the assumption that $\ker \ga  =\{0\}$, i.e., $\ga v=0$ implies $v =0$, and that the eigenvalues are positive and strictly decreasing, i.e., $\g_1>\g_2>\cdots>0$
\citep{CMS07, HH07}. 
The former ensures that the underlying function space $\HH$ is visible through the random function $X$,
while the latter assumption helps us to exclude trivial cases, for example, when the random element $X$ is supported only in a finite-dimensional subspace of $\HH$.
 
We can define a square-root inverse operator by using the spectral decomposition  as $\ga^{-1/2} \equiv \sum_{k=1}^\infty \g_k^{-1/2} (\phi_k \otimes \phi_k)$. 
Then, when $\HH$ is the space $L^2([0,1])$ of square integrable functions from $[0,1]$ to $\R$, 
where the inner product is defined by integration $\langle f, g \rangle = \int_0^1 f(t)g(t)dt$ for $f,g \in \HH$,
the set $\HH(\ga) \equiv \{ v \in \HH: \|\ga^{-1/2} v\| <\infty \}$ is an RKHS with reproducing kernel $(s,t) \mapsto \cov[X(s), X(t)]$
if the covariance function $\ga$ is continuous \citep{wahba73}. 
The subset $\HH(\ga)$ is relatively large in $\HH$ in the sense that $\HH(\ga)$ is an infinite-dimensional dense subspace of $\HH$. 
Inspired by this, \cite{YDL25RHD} suggested restricting the directions in the Tukey's halfspace depth to have the RKHS norm $\|v\|_{\HH(\ga) } \equiv \|\ga^{-1/2} v\|$ less than the regularization $\ld \in (0,\infty)$ and proposed the RHD for dense functional data as defined next.

\begin{defn}
	The \emph{regularized halfspace depth (RHD)} of a (fully observed) functional element $x \in \HH$ is defined as 
	\begin{align} \label{eqRHDdense}
		D^{\mathrm{dense}}_\ld(x) \equiv \inf_{v \in \vv_\ld} \pr( \langle X - x, v \rangle \geq 0 ),
	\end{align}
	where $\ld \in (0, \infty)$ is a regularization and the direction set $\vv_\ld$ is defined by
	\begin{align} \label{eqProjDir}
		\vv_\ld \equiv \{ v \in \HH: \|v\|=1, \|\ga^{-1/2}v\| \leq \ld \}. 
	\end{align}
\end{defn}

The role of the restricted direction set $\vv_\ld$ is to regularize the original Tukey halfspace depth in the infinite-dimensional Hilbert space $\HH$. 
Indeed, if one takes the infimum over the full unit sphere $\vv \equiv \{v\in\HH:\|v\|=1\}$, the resulting halfspace depth can degenerate in infinite-dimensional space. 
By contrast, the RHD restricts the admissible directions to $\vv_\ld$ in \eqref{eqProjDir}, 
thereby controlling the influence of higher-order eigenfunctions when computing the infimum of halfspace probabilities. 
Under mild assumptions, this regularization prevents degeneracy while preserving a rich collection of projection directions. 
In particular, $\vv_\ld$ can still contain directions involving infinitely many eigenfunctions, 
although the contribution of higher-order components is controlled. 
Thus, the RHD is able to incorporate information beyond the leading eigenfunctions. 
Moreover, increasing $\ld$ enlarges the admissible direction set and can make the RHD more sensitive to shape variation in functional data,
which is particularly useful in applications such as shape outlier detection. 
We refer to \cite{YDL25RHD} for further discussion of the theoretical and practical benefits of the RHD in \eqref{eqRHDdense}.

In what follows, we assume that $\ld>\g_1^{-1/2}$ is fixed to guarantee the non-emptiness of the direction set $\vv_\ld$ in \eqref{eqProjDir}.


\subsection{Depth definition and basic properties of CRHD} \label{ssec_2_2}

We now extend the RHD in \eqref{eqRHDdense}, originally developed for densely observed functional data, to the setting where depth is evaluated at sparsely and irregularly observed functional observations. 
The key idea is to replace the halfspace probability in the original RHD with a \emph{conditional} halfspace probability, 
conditioning on the sparse observations at which the depth is evaluated. 
This construction assigns a depth value directly to the observed sparse functional values, rather than to a reconstructed version of the underlying function. 
Consequently, unlike reconstruction-based approaches or methods that require a certain degree of observability over the domain,
the proposed CRHD can be applied directly to sparse and irregular functional observations without imposing such additional constraints.

To explain, we restrict our attention to case of $\HH \equiv L^2([0,1])$.
The mean function is then given as $\mu(t) = \eo[X(t)]$ for $t \in [0,1]$. 
To simplify notation, we use the symbol $\ga$ for both covariance operator and function. 
Namely, the covariance function of $X$ is a bivariate function on $[0,1]^2$ defined as $\ga(s,t) = \cov[X(s), X(t)]$ for each $s,t \in [0,1]$. 
The covariance operator and function are then related as
$(\ga f)(t) = \int \ga(t,s) f(s)  ds$ for $t \in [0,1]$ and $f \in \HH$. 

We consider practical situations where functional data are sparsely observed at a few random times with additive measurement errors \citep{YMW05a, YMW05b, Yao07, LH10, ZW16, DMT18}.
To formally describe the CRHD evaluation at a sparse functional datum,
let $X_0$ be an independent copy of the underlying function $X$.
Since we may not be able to fully observe $X_0$, we suppose that noise-contaminated values $\{\tilde{X}_{0j}\}_{j=1}^{n_0}$ of $X_0$ are observed at $n_0$ time measurements $\{T_{0j}\}_{j=1}^{n_0}$ as follows: 
\begin{align*}
	\tilde{X}_{0j} = X_0(T_{0j}) + \e_{0j}, \quad j=1,\dots,n_0.
\end{align*}
Here, we assume 
the random times $\{T_{0j}\}_{j=1}^{n_0}$ follow the distribution of a random variable $T$ on [0,1] with continuous distribution function $F_T$,
the additive random errors $\{\e_{0j}\}_{j=1}^{n_0}$ are iid copies of a random variable $\e$ with mean zero and variance $\s_\e^2$, and
the error $\{\e_{0j}\}_{j=1}^{n_0}$ are independent of $X_0$.
That is, we observe the pairs $\{(T_{0j}, \tilde{X}_{0j})\}_{j=1}^{n_0}$,
which are identically distributed as $(T, \tilde{X})$ where $\tilde{X} = X(T) + \e$. 
Then, the pair $\bm{\xx}_0 = (\bm{T}_0, \bm{X}_0)$ denotes the sparse functional datum observed from the true function $X_0$,
where $\bm{T}_0 = [T_{01}, \dots, T_{0n_0}]^\top$ and $\bm{X}_0 = [\tilde{X}_{01}, \dots, \tilde{X}_{0n_0}]^\top$  are the vectors of measurement times and of the observed data, respectively.
In this setup, 
by using the conditional distribution of the underlying true function $X_0$  given the corresponding sparse observation $\bm{\xx}_0$,
we propose a novel depth evaluated at $\bm{\xx}_0$ as follows.

\begin{defn}
	We define the \emph{conditional regularized halfspace depth (CRHD)} $D_\ld(\bm{\xx}_0)$ of sparse functional datum $\bm{\xx}_0$ by
	\begin{align} 
		D_\ld(\bm{\xx}_0)
		& \equiv \inf_{v \in \vv_\ld} \pr ( \langle X - X_0, v \rangle \geq 0 | \bm{\xx}_0), \label{eqSRHD1} 
	\end{align}
	where $\ld \in (0,\infty)$ is a regularization and the direction set $\vv_\ld$ is given in \eqref{eqProjDir}.
\end{defn}

It is important to note that, although the direction set $\vv_\ld$ remains rich enough to contain directions formed by infinite linear combinations of eigenfunctions, 
its regularization constraint controls the contribution of higher-order eigenfunctions when computing the depth in \eqref{eqSRHD1}. 
This control prevents the infimum of the conditional halfspace probabilities from collapsing to zero,
thereby leading to the nondegeneracy of the CRHD in \eqref{eqSRHD1} under mild conditions. 
The precise statement is given below, and its proof is deferred to the supplement.

\begin{thm}[Non-degeneracy of the CRHD] \label{thmNondege}

    Suppose that the following distributional condition on $X$ holds: 
    \begin{align} \label{condNondege}
        \sup_{v \in \vv_\ld} \pr \left( \langle X-\mu, v \rangle / \|\ga^{1/2} v\| \leq t \right) <1, \quad \forall t \in (0,\infty).
    \end{align}
    Then, for any sparse functional datum $\bm{\xx}_0$, we have $D_\ld(\bm{\xx}_0)>0$ almost surely.
\end{thm}

The distributional condition in \eqref{condNondege} is relatively mild and indeed includes various random elements.
In particular, examples encompass elliptically-distributed functional data \citep{BB09, BBT14} beyond the Gaussian elements.
Namely, this condition includes random functions whose projections follow well-known distributions such as normal, t, logistic, or Laplace distributions, 
which may not have finite higher moments. 

The theoretical properties of data depth, initially discussed by \cite{Liu90} and later rigorously postulated by \cite{ZS00}, include affine invariance, maximality at center, monotonicity relative to the deepest point, and vanishing at infinity. 
However, these properties cannot be directly extended to functional depth due to its infinite dimensionality. 
\cite{NB16,GN17} comprehensively discussed and postulate dense functional versions of depth properties.
The challenge is further compounded by the lack of linear structures for sparse functional data, making it difficult to find their sparse versions.
To overcome this, we borrow the Hilbert space structure in the underlying function space to establish the desirable theoretical properties of the CRHD. 
To our knowledge, such results have not been previously developed for sparse functional depth. 
These properties are provided in \autoref{thmProperties} and their proofs are given in the supplement.

\begin{thm} \label{thmProperties}
	\hfill
	\begin{enumerate}[(a)]
		\item (Isometry invariance)
		Let $A:\HH \to\HH$ be a bounded linear operator and $b \in \HH$. 
		We define
		$\bm{A}\bm{X}_0 + \bm{b} \equiv ((AX)(T_{01})+b(T_{01}), \dots, (AX)(T_{0n}) + b(T_{0n}))^\top$,
		and let $\bm{\yy}_0 \equiv (\bm{T}, \bm{A}\bm{X}_0 + \bm{b})$ denote the sparse functional datum affinely transformed by $A$ and $b$, where $Y_0 \equiv AX_0 + b$ represents the corresponding underlying true function.
		If $A$ is a surjective isometry, 
		then we have $D_\ld(\bm{\yy}_0; P_{AX+b}) = D_{\ld}(\bm{\xx}_0)$ almost surely,
		where $ P_{AX+b}$ denotes the probability distribution of $AX+b$.

		\item 
		(Maximality at center)
		Suppose that $X$ is halfspace symmetric about a unique center $\theta \in \HH$ in the sense that $\pr(X \in H) \geq 1/2$ for each $H \in \hh(\theta)$, 
		where $\hh(x)$ denotes the collection of all closed halfspaces containing $x \in \HH$. 
		If $\ld>\g_1^{-1/2}$ where $\g_1$ is the largest eigenvalue of $\ga$, 
		then for any sparse functional datum $\bm{\xx}_0$, we have $D_\ld(\bm{\xx}_0) \leq D^{\mathrm{dense}}_\ld(\theta)$ almost surely. 
		
		\item (Monotonicity relative to the deepest point)
		Let $\theta \in \HH$ be such that 
		$D^{\mathrm{dense}}_\ld(\theta) \geq D_\ld(\bm{\xx}_0)$ almost surely for any sparse functional datum $\bm{\xx}_0$. 
		Then, for any sparse functional datum $\bm{\xx}_0$ and $\ap \in [0,1]$, it holds that $D_\ld(\bm{\xx}_0) \leq D_\ld(\bm{\theta} + \ap(\bm{\xx}_0 - \bm{\theta}))$ almost surely,
		where $\bm{\theta} = (\theta(T_{01}), \dots, \theta(T_{0n_0}))^\top \in \R^{n_0}$. 
		Here, the addition and scalar multiplication of sparse functional data are defined as in part (a). 
		
		\item (Vanishing at infinity)
		Let $\W$ be any finite-dimensional subspace of $\HH$, 
		where $\Pi_\W$ denote the projection operator onto $\W$. 
		Also, let $\{\bm{\xx}_{0n}\}$ be a sequence of sparse functional data
		and suppose that the corresponding true functions $X_{0n}$
		satisfy
		 $\|\Pi_\W X_{0n}\| \to \infty$ as $n\to\infty$ almost surely.
		Then, we have $D_\ld(\bm{\xx}_{0n}) \to 0$ as $n\to\infty$ almost surely. 
		
	\end{enumerate}
\end{thm}


For deriving the theoretical properties delineated in \autoref{thmProperties},
we leverage the Hilbert space structure of $\HH$.
\autoref{thmProperties}(a) defines the affine transformation $\bm{\yy}_0$ of a sparse functional datum $\xx_0$ as the sparse functional datum corresponding to the affine transformation $Y_0 \equiv AX_0 + b$ of the underlying true function $X_0$. 
This transformation of $X_0$ by a surjective isometry $A$ is indeed affine on $\HH$, as per the Mazur--Ulam theorem \citep{vai03}. 
\autoref{thmProperties}(c) describes a convex combination of two sparse data points in a pointwise manner, akin to the affine transformation in (a). 
The deepest point $\theta \in \HH$ in sparse functional data, as defined in \autoref{thmProperties}(b-c), is based on the original RHD $D_\ld^{\mathrm{dense}}$ for fully observed functional data. 
In \autoref{thmProperties}(d), 
the norm of the underlying true function $X_{0n}$ corresponding to the given sparse functional datum $\bm{\xx}_{0n}$ is employed instead of a norm of $\bm{\xx}_{0n}$.

\subsection{Sample CRHD}  \label{ssec_2_3}


Since the population CRHD  in \eqref{eqSRHD1} is not feasible to compute,
we shift our focus to estimating it using the observed sparse functional data.
Let $X_1, \dots, X_n$ be independently and identically distributed (iid) copies of the underlying function $X$.
The observed sparse functional data $\{\bm{\xx}_i = (\bm{T}_i, \bm{X}_i)\}_{i=1}^n$ then consists of the vectors $\bm{T}_i = [T_{i1}, \dots, T_{in_i}]^\top$ and $\bm{X}_i = [\tilde{X}_{i1}, \dots, \tilde{X}_{in_i}]^\top$ of measurement times and the observed data.
We assume that the random times $T_{ij} \in \ttt$ that follows  $F_T$ and the observed data are generated as
\begin{align} \label{eqObsData}
	\tilde{X}_{ij} = X_i(T_{ij}) + \e_{ij}, \quad j=1,\dots,n_i, i=1, \dots,n,
\end{align}
where the additive errors $\{\e_{ij}\}_{j=1}^{n_i}$ are iid with mean zero and variance $\s_\e^2$, which are independent of $X_i$ for each $i$.
This framework mirrors the setup for the sparse functional datum $\bm{\xx}_0$, where the depth is evaluated.
To employ the sparse observations $\{\bm{\xx}_i\}_{i=1}^n$, 
we introduce an averaged version of the CRHD next.
\begin{defn}
	For $\ld \in (0,\infty)$, the \emph{averaged conditional regularized halfspace depth (ACRHD)} $D_{\mathrm{ave},\ld}(\bm{\xx}_0) $ of sparse functional datum $\bm{\xx}_0$ is defined by
	\begin{align}
		D_{\mathrm{ave},\ld}(\bm{\xx}_0) 
		&
		\equiv \inf_{v \in \vv_\ld} n^{-1} \sum_{i=1}^n \pr ( \langle X_i - X_0, v \rangle \geq 0 | \bm{\xx}_i, \bm{\xx}_0),  \label{eqESRHD}
	\end{align}
	where the direction set $\vv_\ld$ is given in \eqref{eqProjDir}.
\end{defn}

As the conditional distribution plays an important role in our proposed depths, 
its estimation becomes crucial in addition to that of $\vv_\ld$.
As such, to compute the ACRHD in \eqref{eqESRHD}, 
we determine the conditional distribution of the projection $\langle X_i - X_0, v \rangle$ given the sparse functional data $(\bm{\xx}_i,\bm{\xx}_0)$.
We extend the results for the FPC scores, which were exploited by the previous work \citep{YMW05a, YMW05b, GDM21}, to those for general projections, which are stated conditional on all the time points $\{\bm{T}_i\}_{i=0}^n$. 
In particular, we denote the conditional mean and variance of $\bm{X}_i$ given $\bm{T}_i$ by
\begin{align*}
	\bm{\mu}_i & \equiv \eo[\bm{X}_i | \bm{T}_i]  = \left[ \mu(T_{i1}), \dots, \mu(T_{in_i}) \right]^\top, \\
	\bm{\Sigma}_i & \equiv \cov[\bm{X}_i|\bm{T}_i] = \left[ \cov[\tilde{X}_{ij}, \tilde{X}_{ij'}|\bm{T}_i] \right]_{1 \leq j,j' \leq n_i},
\end{align*}
where the $(j,j')$ element in the covariance matrix $\bm{\Sigma}_i$ is given as $\cov[\tilde{X}_{ij}, \tilde{X}_{ij'}|\bm{T}_i] = \s_\e^2 \dt_{jj'} + \ga(T_{ij}, T_{ij'})$.
Here, $\dt_{jj'} = 1$ if $j=j'$ and $\dt_{jj'}=0$ if $j \neq j'$.

Since the marginal mean and variance are given as $\eo[\langle X_i, v \rangle] = \langle \mu, v \rangle$ and $\var[\langle X_i, v \rangle] = \|\ga^{1/2}v\|^2$, 
denoting $(\ga v) (\bm{T}_i) = \left[ (\ga v) (T_{i1}), \dots, (\ga v) (T_{in_i}) \right]^\top$ the cross-covariance (vector) between $\bm{X}_i$ and $\langle X_i, v \rangle$,
the joint mean and variance of $(\bm{X}_i^\top, \langle X_i, v \rangle)^\top$ (conditional on $\bm{T}_i$) are written by
\begin{align*}
	\begin{bmatrix}
		\bm{X}_i \\ \langle X_i, v \rangle
	\end{bmatrix} 
	\sim \left( 
	\begin{bmatrix} \bm{\mu}_i \\ \langle \mu, v \rangle \end{bmatrix},
	\begin{bmatrix} \bm{\Sigma}_i & (\ga v) (\bm{T}_i) \\ (\ga v) (\bm{T}_i)^\top & \|\ga^{1/2} v\|^2 \end{bmatrix}
	\right),
\end{align*}
where $Z \sim (\theta, \Oa)$ represents $\eo[Z] = \theta$ and $\var[Z] = \Oa$ for a general random element $Z$. 
The conditional mean $\eta_i(v) \equiv \eo[\langle X_i, v \rangle| \bm{\xx}_i]$ and variance $\Psi_i(v) \equiv \var[\langle X_i, v \rangle| \bm{\xx}_i]$ of the projection $\langle X_i, v \rangle$ given $\bm{\xx}_i  = (\bm{T}_i, \bm{X}_i)$ are then 
\begin{align} 
	\eta_i(v) = \eta(v|\bm{\xx}_i; \mu, \ga, \s_\e^2) & = \langle \mu, v \rangle + \{(\ga v) (\bm{T}_i)\}^\top \bm{\Sigma}_i^{-1} (\bm{X}_i - \bm{\mu}_i), \label{eqCDmean}\\
	\Psi_i(v) = \Psi(v|\bm{\xx}_i; \ga, \s_\e^2)  & = \|\ga^{1/2} v\|^2 - \{(\ga v) (\bm{T}_i)\}^\top \bm{\Sigma}_i^{-1} \{(\ga v) (\bm{T}_i)\}. \label{eqCDvar}
\end{align}
These conditional mean and variance forms hold for $(X_0, \bm{\xx}_0)$ by replacing $i$ with $0$ in the sub-index, 
which yields the conditional mean and variances of the projection difference $\langle X_i - X_0, v \rangle$ given the sparse functional data $\bm{\xx}_i, \bm{\xx}_0$ as
\begin{align*}
	\langle X_i - X_0, v \rangle | \bm{\xx}_i, \bm{\xx}_0
	\sim \left(  \eta_i(v) - \eta_0(v), \Psi_i(v) +  \Psi_0(v)
	\right).
\end{align*}

To obtain a closed-form expression for the relevant conditional distribution, 
we adopt a Gaussian working model in which the FPC scores of $X$ and the measurement errors are jointly Gaussian. 
Under this model, $\langle X_i, v \rangle$ and $\bm{X}_i$ are jointly Gaussian, which implies that
\begin{align*}
	\langle X_i - X_0, v \rangle | \bm{\xx}_i, \bm{\xx}_0
	\sim \nd \left(  \eta_i(v) - \eta_0(v), \Psi_i(v) +  \Psi_0(v)
	\right),
\end{align*}
where $\nd$ generically denotes the normal distribution in a finite-dimensional space. 
Consequently, under the Gaussian working model, the ACRHD $D_{\mathrm{ave},\ld}$ defined in \eqref{eqESRHD} admits the closed-form representation
\begin{align} \label{eqESRHDgaussian}
	D_{\mathrm{ave},\ld}(\bm{\xx}_0) 
	= \inf_{v \in \vv_\ld} n^{-1} \sum_{i=1}^n \left\{ 1 - \Phi \left( \eta_0(v) - \eta_i(v) \over  \{\Psi_0(v) + \Psi_i(v)\}^{1/2}
	\right) \right\},
\end{align}
where $\Phi$ denotes the cumulative distribution function of the standard normal distribution.

\begin{rem} \label{remGau}
	Gaussian assumptions are commonly imposed in sparse functional data analysis \citep{YMW05a, YMW05b, DMT18, GDM21, ZM25}, 
	and also adopted in some settings with densely observed functional data \citep{PKM10, QGJ19}. 
	The role of Gaussianity in the present work is different. 
	The population definition of the proposed depth in \eqref{eqSRHD1} does not require Gaussianity; 
	it is defined through conditional halfspace probabilities and is therefore meaningful beyond the Gaussian setting. 
	We invoke a Gaussian working model only at the estimation stage, 
	in order to obtain tractable closed-form expressions for the conditional distributions needed to compute the ACRHD in \eqref{eqESRHDgaussian}. 
	In this sense, Gaussianity is not a modeling requirement for the depth itself, 
	but a practical device for estimating it from sparse and noisy observations. 
	Our simulation results in \autoref{sec4} suggest that the resulting depth values and their associated rank-based procedures remain stable under non-Gaussian data-generating mechanisms. 
	Developing a general estimation strategy without the Gaussian working model would require estimating the conditional distribution of a projection given sparse functional observations, 
	which is a challenging problem that, to our knowledge, has not been systematically studied in the literature. 
	We leave this direction for future research.
	We return to this point in \autoref{secConclusion}, 
	where we discuss conditional-distribution estimation for sparse functional observations as a broader future research direction.
\end{rem}

Let $\hat{\mu}$, $\hat{\ga}$, and $\hat{\s}_\e^2$ respectively denote the estimators of $\mu$, $\ga$, and $\s_\e^2$ at hand;
the estimation of these parameters will be detailed soon.
We substitute these estimators $\hat{\mu}, \hat{\ga}, \hat{\s}_\e^2$ into the expressions for the estimated conditional mean and variance as follows:
\begin{align}
	\hat{\eta}_i(v) = \hat{\eta}(v|\bm{\xx}_i; \hat{\mu}, \hat{\ga}, \hat{\s}_\e^2) 
	& \equiv  \langle \hat{\mu}, v \rangle + \{(\hat{\ga} v) (\bm{T}_i)\}^\top \hat{\bm{\Sigma}}_i^{-1} (\bm{X}_i - \hat{\bm{\mu}}_i), \label{eqCDmean_est}\\
	\hat{\Psi}_i(v) = \hat{\Psi}(v|\bm{\xx}_i; \hat{\ga}, \hat{\s}_\e^2)
	& \equiv \|\hat{\ga}^{1/2} v\|^2 - \{(\hat{\ga} v) (\bm{T}_i)\}^\top \hat{\bm{\Sigma}}_i^{-1} \{(\hat{\ga} v) (\bm{T}_i)\} \label{eqCDvar_est},
\end{align}
where $\hat{\bm{\mu}}_i  = [ \hat{\mu}(T_{i1}), \dots, \hat{\mu}(T_{in_i}) ]^\top$ and $\hat{\bm{\Sigma}}_i  = [ \hat{\s}_\e^2 \dt_{jj'} + \hat{\ga}(T_{ij}, T_{ij'}) ]_{1 \leq j,j' \leq n_i}$.
The last unknown quantity in \eqref{eqESRHDgaussian} to be estimated is the projection direction set $\vv_\ld$ given in \eqref{eqProjDir}. 
To define the sample projection direction set, let $\hat{\g}_j$ and $\hat{\phi}_j$ denote the eigenvalue and eigenfunction of the estimated covariance $\hat{\ga}$, respectively.
The estimated projection direction set is then obtained by replacing the population quantities with their estimators:
\begin{align} \label{eqProjDir_est}
	\hat{\vv}_{\ld,K} \equiv \{ \hat{v} \in \mathrm{span} \{\hat{\phi}_1, \dots, \hat{\phi}_K \} : \|\hat{v}\|=1, \|\hat{\ga}_K^{-1/2} \hat{v}\| \leq \ld \},
\end{align}
where $\hat{\ga}_K^{-1/2} \equiv \sum_{k=1}^K \hat{\g}_k^{-1} (\hat{\phi}_k \otimes \hat{\phi}_k)$ is a finite sample approximation of $\ga^{-1/2} \equiv \sum_{k=1}^\infty \g_k^{-1/2} (\phi_k \otimes \phi_k)$ up to the truncation level $K$.
We define the \emph{sample ACRHD} by
\begin{align} \label{eqESRHDsample}
	\hat{D}_{\mathrm{ave},\ld}(\bm{\xx}_0) 
	= \inf_{\hat{v} \in \hat{\vv}_{\ld, K}} n^{-1} \sum_{i=1}^n \left\{ 1 - \Phi \left(
	\hat{\eta}_0(\hat{v}) - \hat{\eta}_i(\hat{v}) \over \{\hat{\Psi}_0(\hat{v}) + \hat{\Psi}_i(\hat{v})\}^{1/2}
	\right) \right\},
\end{align}
using the estimators $\hat{\eta}_i$, $\hat{\Psi}_i$, and $\hat{\vv}_{\ld,K}$ constructed respectively in \eqref{eqCDmean_est}, \eqref{eqCDvar_est}, and \eqref{eqProjDir_est}.

We emphasize that the estimators $\hat{\mu}$, $\hat{\ga}$, and $\hat{\s}_\e^2$ are used directly in the construction of the sample ACRHD in \eqref{eqESRHDsample}, rather than to reconstruct the underlying trajectories. 
This is a key distinction from two-stage approaches, such as those in \cite{LW11, LQ21, SL21, QG22}, which first use such estimators to predict the underlying curves and then evaluate depth values on the reconstructed data. 
By avoiding this reconstruction step, the proposed sample ACRHD is designed to retain the uncertainty induced by sparse and noisy observations more directly. 
In contrast, two-stage approaches may lose part of this uncertainty through the smoothing or averaging inherent in trajectory prediction. 
As demonstrated numerically in \autoref{sec4}, 
the proposed CRHD for sparse functional data outperforms the standard two-stage approach that applies the original RHD in \eqref{eqRHDdense} to reconstructed curves across a range of scenarios.

\begin{rem}
	Although Equation~\eqref{eqESRHDgaussian} provides a closed-form expression for the ACRHD under the Gaussian working model, 
	its computation still requires estimation of the unknown model quantities: 
	the mean function $\mu$, covariance $\ga$, and noise variance $\s_\e^2$. 
	For this purpose, we adopt the PACE method, which uses local linear smoothing to estimate these quantities and is widely used in sparse functional data analysis \citep{YMW05a, YMW05b, LH10, ZW16, DMT18, GDM21}. 
	This method is also conveniently implemented in the \texttt{fdapace} package \citep{fdapace}. 
	Other recent estimation techniques \citep[e.g.,][]{Kraus15, KL20, LW22, ZM25} could also be incorporated in principle, provided that they yield sufficiently accurate estimates of the required quantities. 
	Throughout the following discussion, however, $\hat{\mu}$, $\hat{\ga}$, and $\hat{\s}_\e^2$ denote the estimators of $\mu$, $\ga$, and $\s_\e^2$, respectively, obtained from the PACE method.
\end{rem}

The sample CRHDs in \eqref{eqESRHDsample} and \eqref{eqSRHDsample} involve two tuning parameters: 
the regularization parameter $\ld \in (0,\infty)$ and the truncation level $K$.

The regularization parameter $\ld$ is the more influential of the two. 
It is a user-specified parameter that should be chosen according to the purpose of the depth analysis. 
When functional observations differ mainly in magnitude, a relatively small value of $\ld$ may be sufficient. 
In contrast, when differences in shape are of primary interest, a larger value of $\ld$ may be more effective, although such differences can be difficult to recover from sparse and noisy observations. 
If a task-specific criterion is available, such as the misclassification error in classification problems, $\ld$ can be selected to optimize that criterion. 
In the absence of such a criterion, we recommend as a default choice using quantile levels of the empirical RKHS norms, as described in \autoref{alg1}.

The truncation level $K$ is of secondary importance, but it should be chosen sufficiently large so that the sample direction set $\hat{\vv}_{\ld,K}$ in \eqref{eqProjDir_est} provides an adequate approximation to the population direction set $\vv_\ld$ in \eqref{eqProjDir}. 
A common and practical choice is based on the fraction of variance explained (FVE), which selects $K$ so that the leading FPCs explain a prescribed proportion of the sample variation \citep{YMW05a, DMT18, GDM21}. 
Specifically, for a given threshold $\rho \in (0,1)$, such as $\rho=0.85$, we choose the smallest $K$ whose FVE is no less than $\rho$. 
Once $K$ is sufficiently large, the performance of the CRHDs is typically relatively insensitive to its exact value.

\begin{rem}
	
	As with the ACRHD in \eqref{eqESRHDgaussian}, the Gaussian working model yields a closed-form expression for the original CRHD in \eqref{eqSRHD1}:
	\begin{align} \label{eqSRHDgaussian}
		D_\ld(\bm{\xx}_0)
		= \inf_{v \in \vv_\ld} \left\{ 1-\Phi \left( \theta_0(v) - \langle \mu, v \rangle \over ( \Psi_0(v) + \|\ga^{1/2}v\|^2)^{1/2} \right) \right\},
	\end{align}
	where $\theta_0(v)$ and $\Psi_0(v)$ are given in \eqref{eqCDmean}--\eqref{eqCDvar}, respectively. 
	By plugging the corresponding estimators $\hat{\theta}_0$ and $\hat{\Psi}_0$ from \eqref{eqCDmean_est}--\eqref{eqCDvar_est} into this expression, we define the \emph{sample plug-in CRHD} (PCRHD) as
	\begin{align}  \label{eqSRHDsample}
		\hat{D}_{\mathrm{plug},\ld}(\bm{\xx}_0)
		= \inf_{\hat{v} \in \hat{\vv}_{\ld,K}} \left\{ 1-\Phi \left( \hat{\theta}_0(\hat{v}) - \langle \hat{\mu}, \hat{v} \rangle \over ( \hat{\Psi}_0(\hat{v}) + \|\hat{\ga}^{1/2}\hat{v}\|^2)^{1/2} \right) \right\}.
	\end{align}
	This plug-in version effectively applies the Gaussian working model at the population level, through the unconditional distribution of the underlying process $X$. 
	Consequently, PCRHD may be more sensitive to violations of Gaussianity. 
	In contrast, the ACRHD in \eqref{eqESRHDgaussian} uses the Gaussian working model only to compute conditional distributions associated with the observed sparse trajectories, and then averages the resulting pairwise conditional halfspace probabilities over the sample. 
	This more localized use of the Gaussian working model can make ACRHD less sensitive to non-Gaussian features of the underlying process. 
	Although PCRHD provides a natural alternative to ACRHD, 
	our simulation results in \autoref{sec4} indicate that ACRHD generally performs better.
\end{rem}

\section{Computation by random projections}  \label{sec3}

This section introduces a practical computation method for the CRHDs.
We adopt a random projection approach, which has been widely used in depth literature \cite[Section~6]{Dyck2004}.

Since the sample direction set $\hat{\vv}_{\ld,K}$ in \eqref{eqProjDir_est} is $K$-dimensional, 
the sample CRHDs can be re-expressed using directions in the $K$-dimensional Euclidean space $\R^K$, facilitating its subsequent approximation.
We denote the direction $\hat{v} = \sum_{k=1}^K \hat{a}_k \hat{\phi}_k \in \hat{\HH}_K \equiv \mathrm{span}\{ \hat{\phi}_1, \dots, \hat{\phi}_K \}$ by a linear combination of the first $K$ sample eigenfunctions $\{\hat{\phi}_k\}_{k=1}^K$ with the $K$-dimensional coefficient vector $\hat{\bm{a}} = (\hat{a}_1, \dots, \hat{a}_K)^\top \in \R^K$. 
For further development, let 
\begin{align} \label{eqProjDirA} 
	\hat{\aaa}_{\ld,K} \equiv \{ \hat{\bm{a}} \in \R^K: \|\hat{\bm{a}}\|_{\R^K} = 1, \|\hat{\bm{\rr}}_K^{-1/2} \hat{\bm{a}}\|_{\R^K} \leq \ld\}
\end{align}
be the intersection of the unit sphere and a ellipsoid centered at the origin in $\R^K$, where $\hat{\bm{\rr}}_K \equiv \mathrm{diag}(\hat{\g}_1, \dots, \hat{\g}_K)$ is the diagonal matrix with components being the first $K$ sample eigenvalues. 
It is worthwhile noting that there is one-to-one correspondence 
between the estimated direction set $\hat{\vv}_{\ld,K}$ in \eqref{eqProjDir_est} and this $K$-dimensional direction set $\hat{\aaa}_{\ld,K}$ in \eqref{eqProjDirA} as 
\begin{align} \label{eq11map}
	\hat{\aaa}_{\ld, K} \ni \hat{\bm{a}} = (\hat{a}_1, \dots, \hat{a}_K)^\top \mapsto \hat{v} = \sum_{k=1}^K \hat{a}_k \hat{\phi}_k \in \hat{\vv}_{\ld,K}
\end{align}
since $\|\hat{v}\| = \|\hat{\bm{a}}\|_{\R^K}$ and $\|\hat{\ga}_K^{-1/2}\hat{v}\| = \|\hat{\bm{\rr}}_K^{-1/2} \hat{\bm{a}}\|_{\R^K}$.
Through the bijective map in \eqref{eq11map}, the conditional mean and variance estimators can be computed by using $K$-dimensional directions as
\begin{align}
	\hat{\theta}_{K,i}(\hat{\bm{a}})
	= \hat{\theta}_K(\hat{\bm{a}}|\bm{\xx}_i; \hat{\mu}, \hat{\ga}, \hat{\s}_\e^2)
	& \equiv (\hat{\bm{\Pi}}_K \hat{\mu})^\top \hat{\bm{a}} + \hat{\bm{a}}^\top \hat{\bm{\rr}}_K \hat{\bm{\Phi}}_{iK}^\top \hat{\bm{\Sigma}}_i^{-1} (\bm{X}_i - \hat{\bm{\mu}}_i), \label{eqCDmean_est_a} \\
	\hat{\Psi}_{K,i}(\hat{\bm{a}})
	= \hat{\Psi}_K(\hat{\bm{a}}|\bm{\xx_i}; \hat{\ga}, \hat{\s}_\e^2)
	& \equiv \|\hat{\bm{\rr}}_K^{1/2} \hat{\bm{a}}\|_{\R^K}^2 - \hat{\bm{a}}^\top \hat{\bm{\rr}}_K \hat{\bm{\Phi}}_{iK}^\top \hat{\bm{\Sigma}}_i^{-1} \hat{\bm{\Phi}}_{iK} \hat{\bm{\rr}}_K \hat{\bm{a}}, \label{eqCDvar_est_a}
\end{align}
where the operator $\hat{\bm{\Pi}}_K:\HH \to \R^K$ with finite rank is defined as 
\begin{align} \label{eqProjOp}
	\hat{\bm{\Pi}}_K f \equiv \left[ \langle f, \hat{\phi}_1 \rangle, \dots, \langle f, \hat{\phi}_K \rangle \right]^\top, \quad f \in \HH
\end{align}
and $\hat{\bm{\Phi}}_{iK} = [ \hat{\phi}_k(T_{ij}) ]_{1 \leq j \leq n_i, 1 \leq k \leq K}$ is an $n_i \times K$ matrix with property that
\begin{align*}
	(\hat{\ga} \hat{v}) (\bm{T}_i) 
	= \left[ \sum_{k=1}^K \hat{a}_k \hat{\g}_k \hat{\phi}_k(T_{ij}) \right]_{1 \leq j \leq n_i}
	= \hat{\bm{\Phi}}_{iK} \hat{\bm{\rr}}_K \hat{a}.
\end{align*}
The sample ACRHD in \eqref{eqESRHDsample} can be written by using \eqref{eqCDmean_est_a}-\eqref{eqCDvar_est_a} as
\begin{align} \label{eqESRHDsample_a}
	\hat{D}_{\mathrm{ave},\ld}(\bm{\xx}_0) 
	= \inf_{\hat{\bm{a}} \in \hat{\aaa}_{\ld, K}} n^{-1} \sum_{i=1}^n \left\{ 1 - \Phi \left( 
	\hat{\theta}_{K,0}(\hat{\bm{a}}) - \hat{\theta}_{K,i}(\hat{\bm{a}}) \over \{\hat{\Psi}_{K,0}(\hat{\bm{a}}) + \hat{\Psi}_{K,i}(\hat{\bm{a}})\}^{1/2}
	\right) \right\}
\end{align}
since the conditional quantities in \eqref{eqCDmean_est_a}-\eqref{eqCDvar_est_a} match those in \eqref{eqCDmean_est}-\eqref{eqCDvar_est} through the one-to-one correspondence in \eqref{eq11map}.
We draw a sufficiently large number $L$ of random vectors $\tilde{\aaa}_{\ld,K,L} = \{ \bm{a}_l = (\hat{a}_{l1}, \dots, \hat{a}_{lK})^\top \}_{l=1}^L$ independently and identically distributed from a continuous distribution $\nu$ supported on $\hat{\aaa}_{\ld,K} \subseteq \R^L$ to construct a finite direction set 
\begin{align*}
	\tilde{\vv}_{\ld,K,L} \equiv \left\{ \tilde{v} = \sum_{k=1}^K \hat{a}_k \hat{\phi}_k: \bm{a} = (\hat{a}_1, \dots, \hat{a}_K)^\top \in \tilde{\aaa}_{\ld,K,L} \right\}.
\end{align*}
We then approximate the sample ACRHD in  \eqref{eqESRHDsample} or \eqref{eqESRHDsample_a} with the minimum halfspace probability over $\tilde{\vv}_{\ld,K,L}$ as
\begin{align} \label{eqESRHDapprox}
	\tilde{D}_{\mathrm{ave},\ld,L}(\bm{\xx}_0) 
	= \min_{1 \leq l \leq L} n^{-1} \sum_{i=1}^n \left\{ 1 - \Phi \left( 
	\hat{\theta}_{K,0}(\hat{\bm{a}_l}) - \hat{\theta}_{K,i}(\hat{\bm{a}_l}) \over \{\hat{\Psi}_{K,0}(\hat{\bm{a}_l}) + \hat{\Psi}_{K,i}(\hat{\bm{a}_l})\}^{1/2}
	\right) \right\}.
\end{align}
The sample PCRHD in \eqref{eqSRHDsample} is similarly approximated by
\begin{align} \label{eqSRHDapprox}
	\tilde{D}_{\mathrm{plug},\ld, L}(\bm{\xx}_0) 
	= \min_{1 \leq l \leq L} \left\{ 1-\Phi \left( \hat{\theta}_{K,0}(\hat{\bm{a}}_l) - (\hat{\bm{\Pi}}_K \hat{\mu})^\top \hat{\bm{a}}_l \over ( \hat{\Psi}_{K,0}(\hat{\bm{a}}_l) + \|\hat{\bm{\rr}}_K^{1/2} \hat{\bm{a}}_l\|_{\R^K}^2)^{1/2} \right) \right\}.
\end{align}



We employ a rejection sampling idea to draw the random projections $\tilde{\aaa}_{\ld,K,L} = \{ \bm{a}_l \}_{l=1}^L$. 
Namely, the proposal distribution for $\tilde{\aaa}_{\ld,K,L} = \{ \bm{a}_l \}_{l=1}^L$ is determined by the rejection sampling from the distribution $\mathsf{U}_K = \mathsf{U}(\hat{\rr}_K)$ of the normalized normal vector with mean zero and covariance $\hat{\bm{\rr}}_K$.
We set this non-isotropic distribution $\mathsf{U}_K$ on the unit sphere to improve the acceptance rate.
The approximate CRHDs in \eqref{eqESRHDapprox} and \eqref{eqSRHDapprox} are constructed using only the accepted directions from $\mathsf{U}_K$.
Namely, if a direction $\bm{a}$ is drawn from $\mathsf{U}_K$,
we accept it if $\|\hat{\bm{\rr}}_K^{-1/2} \hat{\bm{a}}\|_{\R^K} \leq \ld$ and reject it otherwise.

In addition, when multiple regularization parameters are considered, our computation of the approximate CRHDs in \eqref{eqESRHDapprox} and \eqref{eqSRHDapprox} maintains the monotonicity of the CRHDs in $\ld$ as the exact CRHDs satisfy: for $\ld, \ld' \in (0,\infty)$ with $\ld \leq \ld'$, we have $\tilde{D}_{\mathrm{ave},\ld,L}(\bm{\xx}_0)  \geq \tilde{D}_{\mathrm{ave},\ld',L}(\bm{\xx}_0)$.
For more details, let $\Ld$ be a finite set of distinct regularization parameters used for evaluating the CRHDs.
We first generate a large number $L_0$ of directions $\tilde{\aaa}_{K, L_0} \equiv \{\bm{a}_l\}_{l=1}^{L_0}$ from $\mathsf{U}_K$.
These directions are then simultaneously used to compute the family $\{ \tilde{D}_{\mathrm{ave},\ld,L}: \ld \in \Ld \}$ of the multiple depth functions under consideration, by using the rejection sampling idea explained above.

\begin{rem}
	In practice, depending on the regularization $\ld \in (0,\infty)$ and the first $K$ sample eigenfunctions $\{ \hat{\g}_k \}_{k=1}^K$,
	the ellipsoid $\{ \bm{a} \in \R^K: \|\bm{\rr}_K^{1/2} \bm{a}\|_{\R^K} \leq \ld \}$  in $\R^K$ may either entirely contain the unit sphere $\SSS^K \equiv \{ \bm{a} \in \R^K: \|\bm{a}\|_{\R^K}=1 \}$ or be disjoint from $\SSS^K$.
	This could lead to the sample CRHD \eqref{eqESRHDsample} or \eqref{eqESRHDsample_a} approaching the degenerate depth or the direction set $\hat{\aaa}_{\ld,K}$ in \eqref{eqProjDirA} being empty. 
	To avoid this, we do not directly select the regularization parameter $\ld \in (0,\infty)$.
	Instead, we set the $u$-th quantile of the generated RKHS norms $\nn_L \equiv \{ \| \hat{\bm{\rr}}_K^{-1/2} \hat{\bm{a}}_l \}_{l=1}^L$, where $u \in (0,1)$ is referred to as the \emph{quantile level} for regularization.
	This approach guarantees that $\ld $ lies within the interval $(\hat{\g}_1^{-1/2}, \hat{\g}_K^{-1/2})$,
	thereby indicating that the generated direction set $\tilde{\aaa}_{\ld,K,L}$ is non-empty $(\ld > \hat{\g}_1^{-1/2})$ and that the resulting ACRHD does not approach degeneracy  $(\ld < \hat{\g}_K^{-1/2})$.
	In our simulation, the CRHDs with $u=0.95$ quantile level generally perform well.
	See \cite{WN25} for similar strategy of using quantile levels.
\end{rem}

\autoref{alg1} provides a detailed computation procedure for the proposed depth functions upon incorporating the issues mentioned here.

\begin{algorithm2e}[b!]
	\small
	\caption{
		Approximate sample (averaged and original) CRHDs 
	}
	\label{alg1}
	\KwData{
		Sparsely observed sample $\{\bm{\xx}_i\}_{i=1}^n$ and depth evaluation point $\bm{\xx}_0$ of sparse functional data
	} \nlnonumber
	\KwTune{
		Either regularization set $\Ld$ or quantile level set $\uu$. 
		Also truncation parameter $K$
	}
	
	\KwResult{
		Approximate sample CRHDs $\tilde{D}_{\mathrm{ave},\ld,L}(\bm{\xx}_0)$ in \eqref{eqESRHDapprox} and $\tilde{D}_{\mathrm{plug},\ld,L}(\bm{\xx}_0)$ in \eqref{eqSRHDapprox} for $\ld \in \Ld$
	}
	\For{$l = 1,\dots, L_0$}{
		Generate $\bm{z}_l \sim \nd(0, \hat{\rr}_K)$ and compute $\hat{\bm{a}}_l = \bm{z}_l/\|\bm{z}_l\|_{\R_K}$
		
	}	
	
	\If{
		$\Ld$ \emph{is unspecified}
	}{
		\For{$u \in \uu$}{
			Add the $u$-quantile of $\{ \| \hat{\rr}_K^{-1/2} \hat{\bm{a}}_l \|_{\R^K} \}_{l=1}^{L_0}$ to $\Ld$
		}
	}

	\For{ $\ld \in \Ld$}{
		
		$\tilde{\aaa}_{\ld, K, L} \leftarrow \emptyset$
		
		\For{$l = 1,\dots,L_0$}{
			\If{$\| \hat{\rr}_K^{-1/2} \hat{\bm{a}}_l \|_{\R^K} \leq \ld$}{
				Add $\hat{\bm{a}}_l$ to $\tilde{\aaa}_{\ld, K, L}$
			}\If{$|\tilde{\aaa}_{\ld, K, L}|=L$}
			
		}
		
		$\tilde{D}_{\mathrm{ave},\ld,L}(\bm{\xx}_0) 
		\leftarrow \min_{\hat{\bm{a}} \in \tilde{\aaa}_{\ld,K,L}} n^{-1} \sum_{i=1}^n \left\{ 1 - \Phi \left( 
		\hat{\theta}_{K,0}(\hat{\bm{a}}) - \hat{\theta}_{K,i}(\hat{\bm{a}}) \over \{\hat{\Psi}_{K,0}(\hat{\bm{a}}) + \hat{\Psi}_{K,i}(\hat{\bm{a}})\}^{1/2}
		\right) \right\}$
		
		$\tilde{D}_{\mathrm{plug},\ld, L}(\bm{\xx}_0) 
		\leftarrow \min_{\hat{\bm{a}} \in \tilde{\aaa}_{\ld,K,L}} \left\{ 1-\Phi \left( \hat{\theta}_{K,0}(\hat{\bm{a}}) - (\hat{\bm{\Pi}}_K \hat{\mu})^\top \hat{\bm{a}} \over ( \hat{\Psi}_{K,0}(\hat{\bm{a}}) + \|\hat{\bm{\rr}}_K^{-1/2} \hat{\bm{a}}\|_{\R^K}^2)^{1/2} \right) \right\}$
		
	}
	
\end{algorithm2e}

\section{Numerical studies for rank-based inference} \label{sec4}

We investigate the practical performance of the proposed CRHDs.
\autoref{ssec_4_1} describes the general simulation settings used throughout the numerical studies, 
while Sections~\ref{ssec_4_2}--\ref{ssec_4_4} present the results for rank recovery and rank-based testing. 
Our main objective is to compare direct depth evaluation via the proposed averaged and plug-in CRHDs with two-stage approaches based on either predicted curves or the leading FPC scores. 
We also compare the ACRHD with the PCRHD, expecting the former to be more stable because it relies on averaged conditional halfspace probabilities rather than a global plug-in approximation. 
Overall, the numerical results demonstrate that the proposed CRHDs provide more effective rankings than two-stage approaches in a variety of sparse functional data settings.


\subsection{Simulation setup and competing depths} \label{ssec_4_1}

The underlying curves are independently and identically constructed based on the following truncated Karhunen--Lo\`{e}ve expansion:
\begin{align} \label{eqKLtrunc}
	X \overset{\mathsf{d}}{=} \mu+\sum_{k=1}^{K^*} \g_k^{1/2} \xi_k \phi_k,
\end{align}
where $K^*=15$ denotes the true number of FPCs. The eigenvalues $\{\g_j\}$ are determined by the polynomial decay rates of the eigengaps $\g_j - \g_{j+1} = 2j^{-a}$ with $\g_1 = 2 \sum_{j=1}^\infty j^{-a}$ for various dacay rates $a$ under consideration; a lower $a$ represents wiggle curves while a higher $a$ corresponds to smooth curves. The eigenfunctions are chosen to be the set of the first $K^*$ Fourier basis functions $\{1, \sin (2\pi t), \cos (2 \pi t), \dots \}$ in $L^2([0,1])$. The mean function $\mu$ will be specified depending on the simulation studies.

We outline the combinations of the distributions for the FPC scores $\xi_k$ in \eqref{eqKLtrunc} and the errors $\e_{ij}$ in \eqref{eqObsData} employed in the simulation studies. The FPC scores are generated as $\xi_k = \xi W_k$ via the latent variable $\xi$ and iid random variables $\{W_k\}$, where $\xi$ and $\{W_k\}$ are independent. Three distribution pairs for $(\xi, W_k)$ are examined here: $\xi =1$ and $W_k \sim \nd(0,1)$, $\xi \sim \nd(0,1)$ and $W_k \sim \nd(0,1)$, and $\xi \sim \mathsf{Unif}(-\sqrt{3},\sqrt{3})$ and $W_k \sim \mathsf{Unif}(-\sqrt{3},\sqrt{3})$, which are labeled as Gaussian, NN, and UU, respectively. 
It is important to note that the latter two \textit{dependent} FPC scores may produce complex functional data and make inference challenging, as investigated in other contexts \cite[e.g.,][]{YDN23RB}. 
The errors follow either a centered chi-square distribution with degree of freedom $\nu = \s^2/2$ or a centered normal distribution with standard deviation $\s_\e$, denoted respectively by $\chi^2$ and Normal, where $\s_\e^2 \in \{0.01, 0.1, 1\}$ denotes the noise level. 
The observed data are then generated by $\tilde{X}_{ij} = X_i(T_{ij}) + \e_{ij}$ as in \eqref{eqObsData}, where the number $n_i$ of observation in each trajectory is drawn from $\{2, \dots, 9\}$ with replacement and the measurement times $T_{ij}$ are uniformly sampled in $[0,1]$. The sample size $n$ is considered differently according to the simulations.
This data generation process will apply to all simulation studies discussed in this section.

All simulation results are visualized by colored line graphs. The red and blue curves in the figures respectively indicates the results by evaluating the averaged and plug-in CRHDs in \eqref{eqESRHDapprox} and \eqref{eqSRHDapprox} at the observed data $\{\bm{\xx}_i\}_{i=1}^n$, which are denoted as ``ACRHD'' and ``PCRHD''. 
The green curves, on the other hand, represent the results by applying the original RHD \eqref{eqRHDdense} to the dense curves predicted from the PACE technique \citep{YMW05a, YMW05b, GDM21}, 
which is considered our main competitor. 
This is referred to as ``TwoStageRHD''. 
Another two-stage approach is also considered,
where the depth values are computed by Tukey's depth \citep{tukey75} for the vector of the first $K$ FPC scores from the PACE technique;
we use a random projection approach \cite[cf.][]{CN08} implemented in \texttt{ddalpha} package as its exact calculation is not feasible with high dimension.
This approach is denoted by ``TwoStageTHD''.
For all the RHD-based methods, we chose the quantile levels $u \in \{0.4, 0.6, 0.8, 0.95\}$,
while all methods equally used the truncation levels $K \in \{2, 4, 6, 8, 10\}$ and the approximation sizes $L = 1000K$.


All simulations are based on Monte Carlo simulations with $M=1000$ iterations, and we provide the averages of evaluation criteria depending on purposes.
For the sake of brevity, we report the results when $\xi_j$ are of NN-type and the errors are chi-square-distributed.
It is worth noting that,
while this data generation is distinctly non-Gaussian, 
the proposed depths actually perform even better than in Gaussian cases.

\begin{rem} \label{remCompare}
	The purpose of the numerical studies is to demonstrate the advantages of direct depth evaluation, as implemented by the proposed CRHDs, over two-stage approaches. 
	To make the main comparison as fair and focused as possible, we concentrate on methods that use the same estimation technique, namely PACE, and the same underlying type of depth, namely halfspace-type depth. 
	This allows us to isolate the effect of direct depth evaluation from differences caused by the choice of reconstruction method or depth measure. 
	We do not claim, however, that other existing sparse functional depths lack merit; their usefulness may depend on the context and objective of a given statistical analysis. 
	The discussion in this remark focuses on the methods proposed by \cite{SL21,Elias23,QDG22}, which represent important advances for analyzing sparsely or irregularly observed functional data.
	
	First, the $\text{MBD}_U$ method of \citet{SL21} relies on the FPCA approach of \citet{GGC13}, implemented through the \texttt{refund::fpca.sc()} function. 
	To our knowledge, this is the first functional depth that explicitly incorporates uncertainty arising from sparse observation schemes. 
	However, in some of our simulation settings, this implementation encounters numerical difficulties; in particular, when the contaminated curves are observed relatively densely, the routine may estimate the noise level to be nearly zero. 
	Moreover, \citet{SL21} provide only an in-sample version of $\text{MBD}_U$, whereas out-of-sample depth evaluation is essential for depth-based rank tests \citep[cf.][]{KW52, CS12}. 
	For these reasons, we exclude $\text{MBD}_U$ from our numerical comparison.
	
	Another competing depth is POIFD, introduced by \citet{Elias23} and later extended by \citet{QDG22}. 
	Its implementation requires functional observations to share common grid points, which can be achieved through interpolation, as described in the original work. 
	However, the original \texttt{fdaPOIFD} package \citep{Elias23} does not support out-of-sample depth evaluation, which is required for the rank test considered in \autoref{ssec_4_4}. 
	We therefore modify the original R function according to \citet[Definition~2.3]{Elias23} to enable out-of-sample POIFD evaluation. 
	This modification allows us to include POIFD based on modified band depth and halfspace depth, denoted by POIFD\_MBD and POIFD\_HD, respectively, as competing methods in \autoref{ssec_4_4}.
	
	Finally, there are many possible ways to construct sparse functional depths by applying existing finite-dimensional or functional depths to reconstructed objects, such as estimated FPC score vectors or predicted curves. 
	However, since the proposed CRHDs are designed to inherit the strengths of finite-dimensional halfspace depth and the RHD for dense functional data, which have been discussed in prior work \citep[e.g.,][]{MM22, YDL25RHD}, we focus on comparing direct depth evaluation via CRHDs with two-stage approaches based on reconstructed data.
\end{rem}

\subsection{Ranking recovery} \label{ssec_4_2}

The primary attribute that the proposed CRHDs are expected to possess is the ability to recover the ranking structure derived from the original RHD values of the underlying trajectories $\{X_i\}_{i=1}^n$.
To describe the simulation procedure, 
let  $\{r_{q,i}\}_{i=1}^n$, $q \in \{\mathrm{true}, \mathrm{ave}, \mathrm{plug}, \mathrm{pred}\}$, denote 
the rankings based on the RHD for  the underlying true curves $\{X_i\}_{i=1}^n$,
based on the averaged and plug-in CRHDs for the sparse data $\{\bm{\xx}_i\}_{i=1}^n$,
and based on the original RHD for the predicted curves (i.e., TwoStageRHD), respectively.
The recovery property is assessed 
by determining the (Spearman) correlations between the rankings $\{r_{q,i}\}_{i=1}^n$ and $\{r_{\mathrm{true},i}\}_{i=1}^n$ for each $q \in \{\mathrm{ave}, \mathrm{plug}, \mathrm{pred}\}$.
Refer to the supplement for the detailed procedure and \cite{SL21, Elias23} for similar simulation studies. 
In this study, we take into account the decay rates $a \in \{2, 3.5, 5\}$ and sample sizes $n \in \{50, 200, 400\}$, but we only report the rank correlations in \autoref{figSRHDrankcorr}, when $a=3.5$, $n=50$ with $K=4$, 
and the quantile levels are $u \in \{0.4, 0.95\}$.

\autoref{figSRHDrankcorr} displays the rank correlations as a function of the noise level. 
Both proposed CRHDs achieve better rank recovery than the two-stage method, 
which does not directly account for the uncertainty induced by noise and sparse observation schemes when computing depth values. 
Among the two proposed methods, the ACRHD generally outperforms the PCRHD. 
Across all scenarios, we observe the following patterns. 
First, non-Gaussianity in either the FPC scores or the measurement errors has little effect on the rank recovery performance of the proposed CRHDs. 
Second, the rank correlations tend to decrease as the noise level $\s_\e^2$ increases or as the regularization parameter $\ld$ becomes larger, where $\ld$ is determined by the quantile level $u$. 
The latter decrease may occur because the original RHD with larger regularization can produce more ties among observations with the smallest depth values. 
Finally, as the sample size $n$ and the decay rate $a$ increase, the ability to recover the target rankings tends to improve across all methods.

\begin{figure}[!b]
	\centering
	\includegraphics[width=0.9\linewidth]{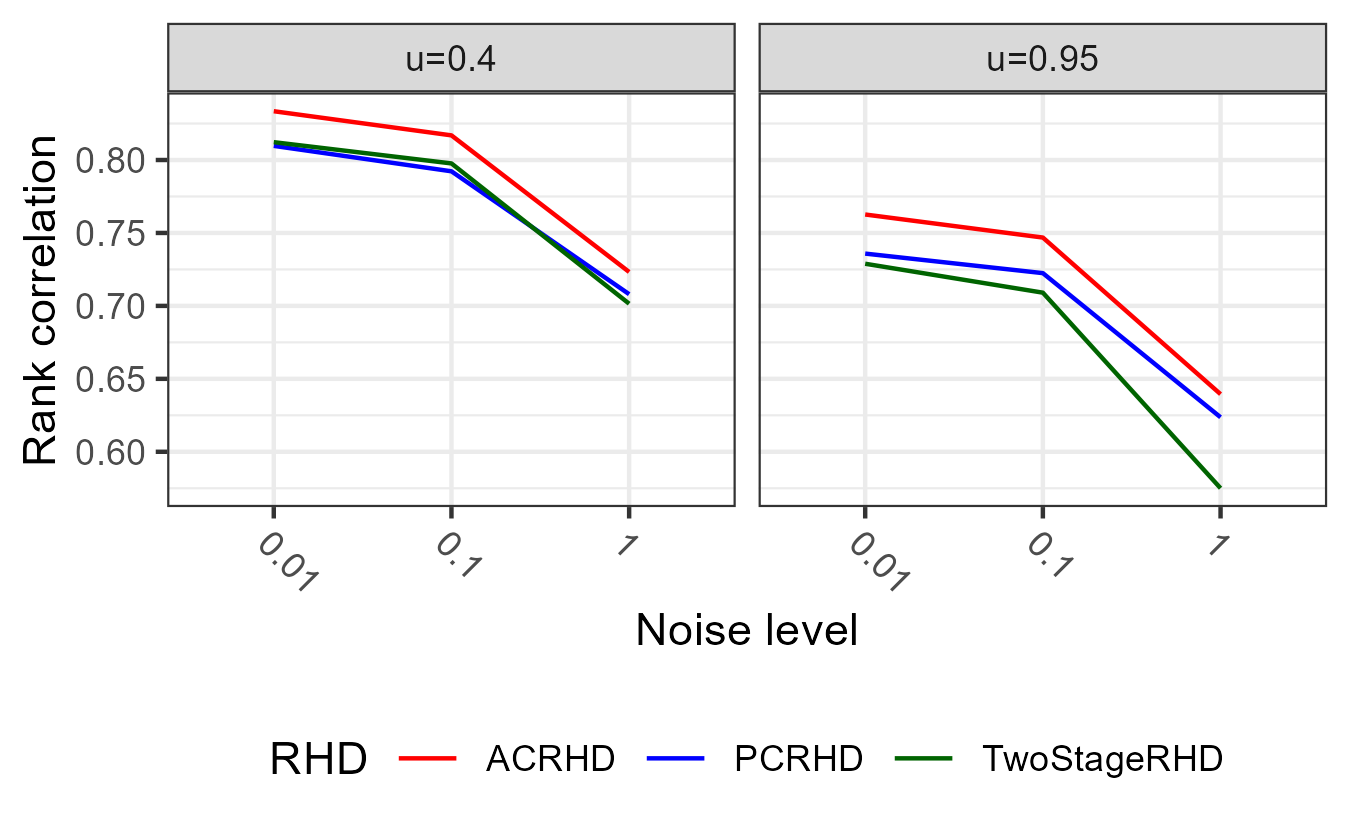}
	\caption{
		Empirical rank correlations of ACRHD, PCRHD, and TwoStageRHD values at sparse observations adding different noise levels with the original RHD values at the true underlying functions.
		The sample size is $n=50$ and the smoothness is determined by $a=3.5$ while the RHDs are constructed with $u \in \{ 0.4, 0.95 \}$ and $K=4$.
		}
	\label{figSRHDrankcorr}
\end{figure}

\subsection{Rank tests} \label{ssec_4_4}

This section examines how effectively the proposed CRHDs and two-stage approaches measure extremeness. 
In rank-based tests, accurately capturing the extremeness of sparse observations is particularly important. 
two-stage approaches can be less effective in this regard because the smoothing or reconstruction step may attenuate features that indicate extremeness, especially under sparse and noisy observation schemes. 
As illustrated by the rank recovery study in \autoref{ssec_4_2}, TwoStageRHD does not recover the RHD-induced rankings of the true underlying curves as accurately as the proposed CRHDs. 
Thus, a key practical distinction between CRHDs and two-stage approaches lies in their ability to preserve extremeness information, and this distinction is expected to be reflected in the performance of rank-based tests.

To provide a more detailed assessment, we evaluate the proposed CRHDs using an existing depth-based rank test, namely the Kruskal--Wallis (KW) test considered by \cite{CS12}, focusing on two-sample problems. 
Specifically, we compute the rank test statistic in \citet[Equation~(10)]{CS12} using the rankings induced by CRHD, and then apply the chi-square approximation described in \citet[Proposition~1]{CS12}. 
The same testing procedure is applied to the competing depths for comparison. 
Further implementation details are provided in the supplement.
In line with \autoref{remCompare}, we do not attempt an exhaustive comparison with the many existing two-sample tests for dense or sparse functional data. 
For comprehensive reviews of functional two-sample inference, we refer to \cite{Wang21, ZW23}; see also \cite{ZSQ25, CLWW25} for recent developments.

To frame the two-sample problems, 
we consider only the cases where the two groups have the same sample size
with total sample sizes $n \in \{100, 200\}$,
indicating the sample sizes $n_0$ and $n_1$ of the control and test groups are equal as $n_0 = n_1 = n/2$.
In all cases described below, we wish to test the null hypothesis $H_0:P_0=P_1$, where $P_0$ and $P_1$ denote the common probability distributions of the control $\{X_{0i}\}_{i=1}^{n_0}$ and test $\{X_{1i}\}_{i=1}^{n_1}$ groups, respectively.
Each (underlying) function in the control group $\{X_{0i}\}_{i=1}^{n_0}$ is independently distributed as $X$ in \eqref{eqKLtrunc} with $\mu^{H_0} = 0$ and $a^{H_0} = 5$.
We consider the following two scenarios for generating the test group: the test group $\{X_{1i}\}_{i=1}^{n_1}$ can differ from the control group $\{X_{0i}\}_{i=1}^{n_0}$ only by either mean or covariance.
Specifically, each function in $\{X_{1i}\}_{i=1}^{n_1}$ still independently follows the Karhunen--Lo\`{e}ve expansion in \eqref{eqKLtrunc} but with different choices of either mean $\mu^{H_1}(t) = \mu_c(t) \equiv  ct$ for $t \in [0,1]$ or decay rate $a^{H_1} = a_c \equiv 5-c$ for $c \in \{0,1,2,3\}$.
The first case of mean difference but with same decay rate $a^{H_0} = 5$ is denoted by MeanDiff,
while the second case of covariance difference but with same mean $\mu^{H_0}=0$ is denoted by CovDiff. 
In both MeanDiff and CovDiff scenarios, $c=0$ renders the null hypothesis $H_0:P_0=P_1$, corresponding to the cases of equal distribution for the two samples $\{X_{0i}\}_{i=1}^{n_0}$ and $\{X_{1i}\}_{i=1}^{n_1}$.
Under these cases, the testing methods are expected to achieve the nominal level.
The sparse observations $\{\bm{\xx}_{0i}\}_{i=1}^{n_0}$ and $\{\bm{\xx}_{1i}\}_{i=1}^{n_1}$ of the two groups are then generated by \eqref{eqKLtrunc} and \eqref{eqObsData}.
We present the empirical rejection rates of the aforementioned KW test based on different depths for sparse functional data. 
The nominal level is chosen as 5\%, 
and for brevity, we only provide the results 
when $\s_\e^2 = 0.1$.
Tthe methods are conducted with quantile level $u=0.95$ (relevant for RHD-related methods) and truncation $K=6$ (applicable to all methods).

\autoref{tb1size} reports the empirical sizes of the KW tests based on different depths. 
The proposed CRHD-based tests, including both the averaged and plug-in versions, closely match the nominal size and are the only procedures that consistently maintain the prescribed level. 
In contrast, the two-stage methods exhibit substantial size distortion, with noticeably inflated false rejection rates. 
The POIFD-based tests also tend to overreject relative to the nominal level, despite showing otherwise reasonable performance in some settings.

\begin{table}[b!] \small
	\centering
	\caption{
		Empirical sizes of the KW tests based on the depths for sparse functional data under consideration
		when the total sample size is $n=100$ and the noise level is $\s_\e^2 = 0.1$.
		The depths are computed by taking $u = 0.95$ and $K=6$. 
		The empirical sizes below the significance level $\alpha = 0.05$ are written in bold,
		while the numbers in the parentheses represent the corresponding standard deviations. 
		}
	\label{tb1size}
	\vspace{0.1in}
	\setstretch{1.3}
	\resizebox{\textwidth}{!}{\begin{tabular}{cccccc} \hline
		ACRHD & PCRHD & TwoStageRHD & TwoStageTHD & POIFD\_HD & POIFD\_MBD \\ \hline
		\textbf{0.049} (0.216) & \textbf{0.032} (0.176) & 0.156 (0.363) & 0.572 (0.495) &
		0.052 (0.222) & 
		0.056 (0.229)\\ \hline
	\end{tabular}}
\end{table}

\autoref{fig_ranktest} presents the empirical rejection rates of the CRHD-based KW tests for different values of $c \in \{0,1,2,3\}$, 
where the gray horizontal lines indicate the nominal 5\% level. 
Because the KW tests based on the two-stage approaches fail to maintain the nominal size, 
we do not report their empirical power in \autoref{fig_ranktest}. 
As expected, the empirical power of both CRHD-based tests increases monotonically with the alternative strength $c$. 
The ACRHD-based rank test tends to be less conservative and therefore achieves higher power than the PCRHD-based test. 
The POIFD-based tests also perform reasonably well, but they appear to be less effective in detecting differences in the second-order structure.

\begin{figure}[!b]
	\centering
	\includegraphics[width=0.89\linewidth]{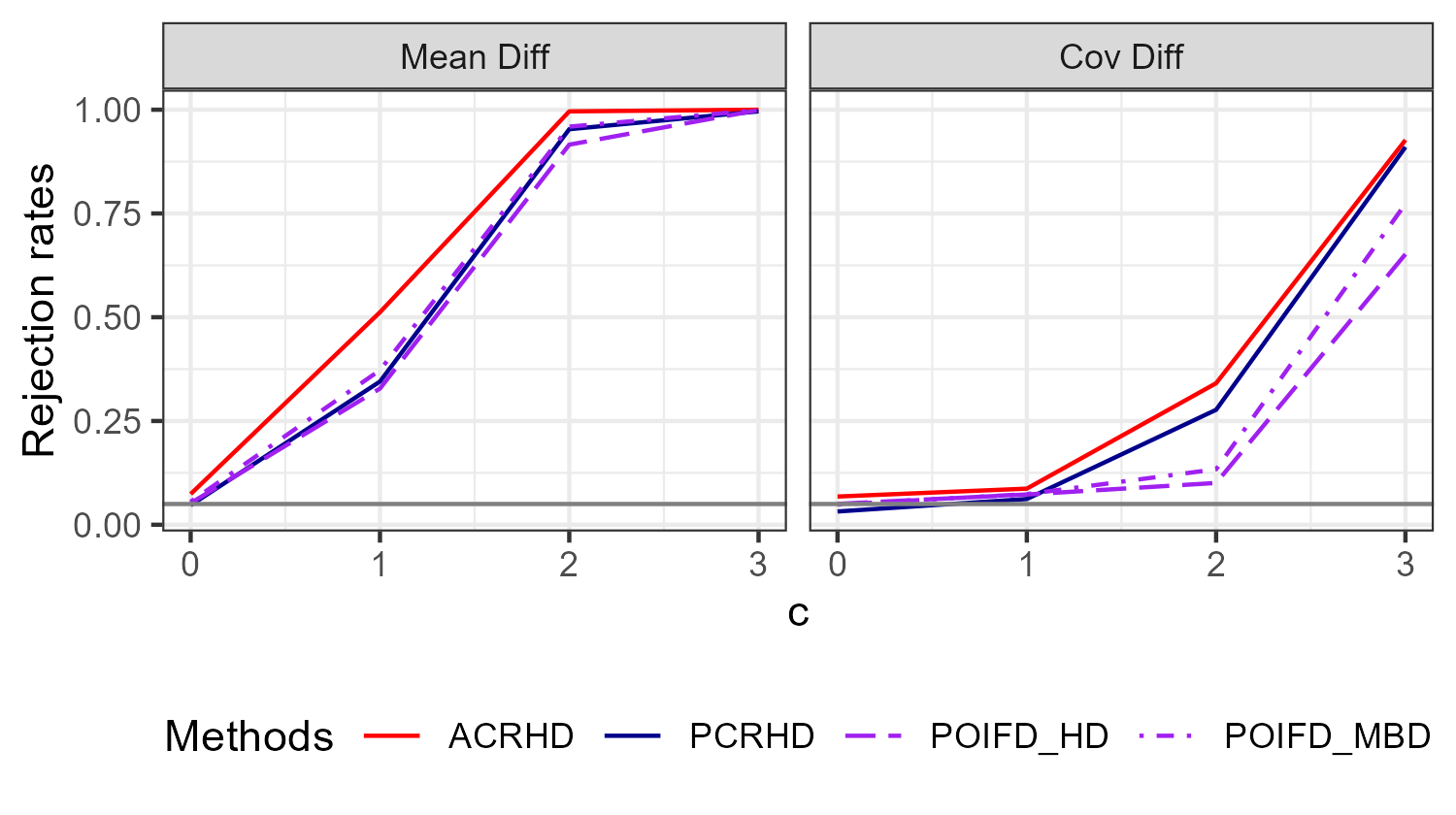}
	\caption{
		Empirical power of the KW tests based on both averaged and plug-in CRHDs as well as the POIFDs using HD (purple dashed lines) and MBD (purple dot-dashed lines),		 
		as the degree $c \in \{0, 1, 2, 3\}$ of the alternative increases; 
		the gray horizontal line represents the nominal size 5\%.
		The total sample size is $n=100$ and the noise level is $\s_\e^2 = 0.1$.
		The depths are computed by taking $u = 0.95$ and $K=6$. 
	}
	\label{fig_ranktest}
\end{figure}

We also observe several additional patterns across the numerical studies, some of which are expected. 
In general, the tests become more powerful as the sample size increases. 
For tests based on ACRHD, PCRHD, and the original RHD, 
a larger regularization parameter improves power in most cases; 
this is consistent with the observation that larger regularization can better emphasize extremeness, 
as discussed for the RHD in \cite{YDL25RHD}. 
Larger noise levels, on the other hand, tend to have an adverse effect on power. 
Finally, Gaussianity of the underlying function $X$ and measurement error $\e$ does not necessarily lead to higher power for the proposed CRHD-based rank tests. 
For example, this can be seen by comparing the Gaussian and NN cases for the random functions when the errors are normal with variance $\s_\e^2=0.1$. 
This finding suggests that the CRHD-based rank tests are relatively stable with respect to the Gaussian working model used in their construction.


\section{Infant growth data analysis} \label{sec5}

This section illustrates the practical benefits of the proposed CRHDs using growth data from the 1988 National Maternal and Infant Health Survey (NMIHS) and its 1991 Longitudinal Followup (LF) \citep{SPK91, SSG98}. 
Although the original dataset contains a variety of variables, we focus on the heights, measured in inches, of boys observed longitudinally over the first two years of life. 
Because doctor visits occurred at irregular times and some responses were missing, 
the resulting growth trajectories are observed sparsely and irregularly. 
Our goal is to examine whether the growth trajectories of normal-birth-weight and low-birth-weight infants, 
defined as those with birth weight no greater than 2500 g, 
follow the same distribution.

This dataset was previously analyzed by \cite{LQ21}, following \cite{LW11, LSLG14}, using a functional-depth-based envelope test. 
Their approach provides a useful visualization tool for interpreting group differences, 
but it relies on predicted curves and is therefore essentially a two-stage procedure. 
In the original dataset, one may expect normal-birth-weight infants to have larger heights overall, as observed in previous analyses. 
However, group differences may also arise from more subtle distributional features, 
such as differences in shape, variability, smoothness, or extremeness of the underlying growth trajectories. 
Such differences are more challenging to detect, especially when the curves are sparsely and irregularly observed.

To focus on this more challenging setting, 
we construct a subset of the original data designed 
to reduce the dominant location difference between the two groups. 
For each subject, we first compute the sample median of the observed heights; 
let $\hat{m}_i$ denote the median of $\{\tilde{X}_{ij}\}_{j=1}^{n_i}$. 
We then retain only subjects whose median height lies in the interval $[24,26]$, and denote the resulting index set by
\[
\ii_{\mathrm{sub}} = \{i: \hat{m}_i \in [24,26]\}.
\]
The resulting dataset contains $n_0=141$ normal-birth-weight infants and $n_1=70$ low-birth-weight infants. 
Because the two groups now overlap substantially in overall magnitude,
any remaining distributional difference is expected to be driven less by location and more by shape-related or extremeness-related features of the sparse growth trajectories,
as depicted in \autoref{fig_rda_growth_curves_subset}.

\begin{figure}[!b]
	\centering
	\includegraphics[width=0.79\linewidth]{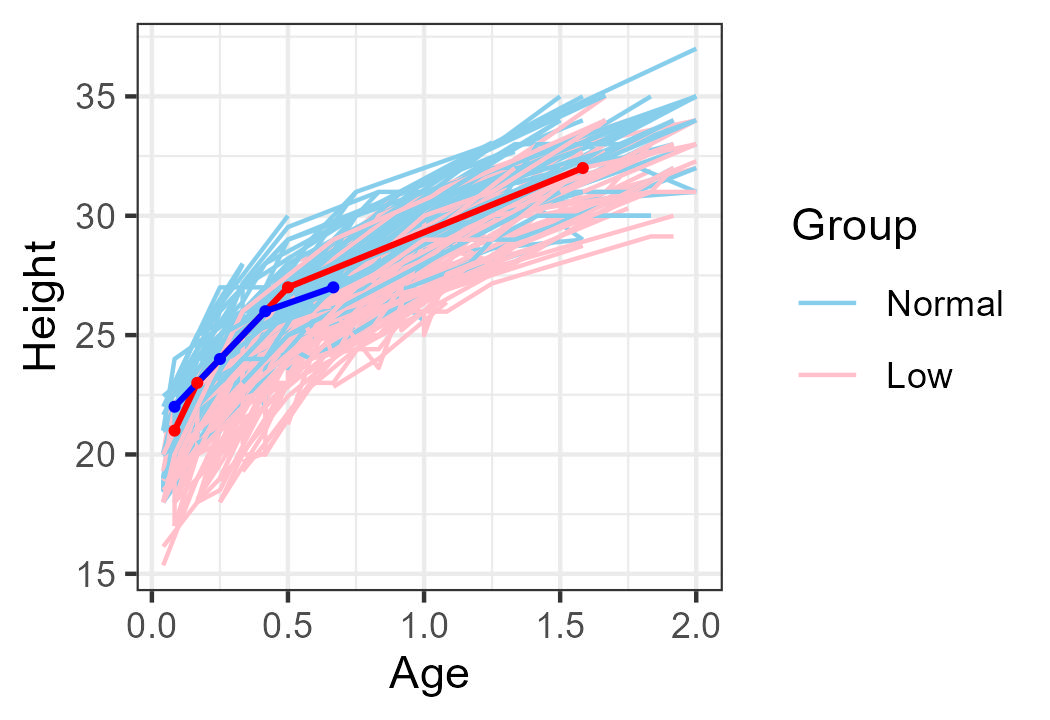}
	\caption{
		Filtered height trajectories from the infant growth data, obtained by retaining subjects whose median observed height lies in $[24,26]$. 
		This filtering reduces the dominant location difference between the normal-birth-weight and low-birth-weight groups, 
		making the comparison more sensitive to shape- and extremeness-related differences. 
		Normal-birth-weight and low-birth-weight infants are shown in blue and red, respectively. 
		The deepest sparse trajectory in each group, computed using ACRHD with $u=0.95$ and $K=6$, is displayed as a thicker and darker curve with observed points.
	}
	\label{fig_rda_growth_curves_subset}
\end{figure}

\begin{table}[b!]
	\centering
	\caption{
		P-values of the CRHD-based KW two-sample tests for the filtered infant growth data. 
		The tests compare the distributions of normal-birth-weight and low-birth-weight growth trajectories 
		after reducing the overall location difference between the two groups. 
		The truncation level and quantile level are fixed at $K=6$ and $u=0.95$, respectively.
	}
	\label{tb_rda}
	\vspace{0.1in}	
	\setstretch{1.3}
	\begin{tabular}{c|cc} \hline
		Methods & ACRHD & PCRHD \\ \hline 
            P-values & $<10^{-5}$ & $1.0\times 10^{-5}$ \\
        \hline 
	\end{tabular}
\end{table}

We apply the KW tests based on the proposed averaged and plug-in CRHDs to this filtered dataset. 
As shown in \autoref{tb_rda}, both tests with $u=0.95$ and $K=6$ reject the null hypothesis at level $0.05$,
providing evidence that the two groups of growth trajectories differ in distribution. 
This result demonstrates the practical value of CRHD-based rankings in a setting where the visual separation between groups is reduced and the difference is likely to involve more complex features than overall magnitude. 
Motivated by the size distortions observed for two-stage rank tests in \autoref{ssec_4_4}, 
we do not report the corresponding results here.

Instead, we compare our findings with the previous analysis of \cite{LQ21}. 
Specifically, we apply their envelope test to the filtered dataset, using the same candidate values $K \in \{2,4,6,8,10\}$ for presmoothing and the extremal depth for envelope construction, as in the original paper. 
For this filtered dataset, however, the envelope tests fail to reject the null hypothesis for all choices of $K$, since the observed median differences are always contained within the corresponding envelopes. 
This contrasts with the results reported by \cite{LQ21} for the original dataset and suggests that the filtered data contain a more subtle form of group difference, not primarily driven by overall magnitude. 
The CRHD-based KW tests are able to detect this difference, illustrating that direct depth evaluation from sparse observations can be beneficial when group differences are expressed through complex shape, variability, or extremeness patterns.

\section{Conclusions and future directions} \label{secConclusion}

This paper makes three main contributions to sparse functional data analysis and the literature on data depth. 
First, we propose the conditional regularized halfspace depth (CRHD), a new depth function designed to evaluate sparsely observed curves directly, without requiring prior reconstruction of the underlying trajectories. 
Second, we show that CRHD induces a useful ranking structure for sparse functional data and can be effectively applied to rank-based hypothesis testing problems. 
Third, our numerical studies demonstrate that direct depth evaluation via CRHD outperforms or remains competitive with existing two-stage approaches, especially by more accurately recovering the RHD-induced rankings of the true underlying trajectories.

As noted in \autoref{remGau}, the proposed depth definition itself does not require Gaussianity; 
the Gaussian working model is used \textit{only} for estimating the conditional distributions involved in its computation.
Accordingly, medians induced by the estimated CRHD may be sensitive to contamination or model misspecification when the Gaussian working model is violated. 
This limitation reflects the intrinsic difficulty of estimating the conditional distribution of a projection $\langle X_0, v \rangle$ given sparse observations $\bm{X}_0$. 
More fundamentally, without suitable distributional assumptions, this conditional distribution may not be identifiable. 
To illustrate this point, let $\mathsf{Rade}(1/2)$ denote the Rademacher distribution, and let $S \sim \mathsf{Rade}(1/2)$ with $\pr(S=1)=\pr(S=-1)=1/2$. 
Consider two processes, $X_0(t)\equiv S$ and $X_1(t)=\mathrm{sign}\{\sin(2\pi(t+U))\}$ for all $t\in[0,1]$, where $U\sim\mathsf{Unif}(0,1)$. 
Then, at any fixed observation time, both $X_0(t)$ and $X_1(t)$ have the same distribution as $S$. 
Thus, under an extremely sparse design with a single observation per curve, the observed data from the two processes have the same marginal distribution. 
However, their projections onto the direction $v(t)\equiv 1$ are different: $\langle X_0,v\rangle=S\sim\mathsf{Rade}(1/2)$, whereas $\langle X_1,v\rangle=0$. 
Indeed, the same observed value, say $X(t_0)=1$, implies $\langle X_0,v\rangle=1$ under the first process but $\langle X_1,v\rangle=0$ under the second. 
Hence, the same sparse observation can correspond to different conditional distributions of the projection, underscoring the need for distributional structure in estimation.

Nevertheless, even in finite-dimensional settings, some useful notions of depth rely on structural features of the underlying distribution, such as its mean or covariance; 
the Mahalanobis depth is a prominent example, and we refer to \cite{MM22} for a related discussion. 
In the same spirit, although the proposed CRHD is model-free at the definitional level, the Gaussian working model provides a practical and tractable route to estimating the conditional distributions required for its computation. 
The resulting procedure is useful for rank-based inference with sparse functional data, particularly because it enables depth evaluation directly from sparse observations without requiring prior trajectory reconstruction. 
We therefore view this work as a foundational step toward depth-based analysis of sparse functional data, offering both a conceptually natural definition and a practically effective implementation. 
More broadly, the preceding discussion points to a new research agenda: developing principled methods for estimating the conditional distribution of a projection given sparse functional observations. 
To the best of our knowledge, this problem has received little attention in the existing literature, despite its fundamental relevance to depth-based inference for sparse functional data. 
While the proposed computational framework may be extended beyond the Gaussian working model, for instance to more general elliptical distributions, we leave these developments for future research.

\section*{Supplementary materials}

The supplement includes the proofs of the theorems, the detailed procedures of the simulation studies, and extra simulation results. 

%
%
%
%
%

%

%
%

%
%
%
%
%
%
\clearpage
\setstretch{1}
\bibliographystyle{dcu}
\bibliography{RHDsparse}

\newpage

\setcounter{page}{0}

\normalsize
\renewcommand{\baselinestretch}{1.0}

\appendix

\begin{center}
	\Large \bf
	Supplement to \\
	``Conditional regularized halfspace depth \\
	for sparse functional data
	and its applications''
\end{center}

\begin{center}
	Hyemin Yeon\footnote{
		Department of Mathematical Sciences, Kent State University, Kent, OH, 44242, USA, Email: hyeon1@kent.edu.
	}, \qquad
	Xiongtao Dai\footnote{
		Division of Biostatistics, University of California, Berkeley, CA, 94720, USA, Email: xiongtao.dai@hotmail.com.
	}, \qquad
	Sara Lopez-Pintado\footnote{
		Department of Health Sciences, Northeastern University, Boston, MA, 02115, USA, Email: s.lopez-pintado@northeastern.edu.
	}
\end{center}

\renewcommand{\thesection}{S\arabic{section}} 
\renewcommand{\thethm}{S\arabic{thm}}
\renewcommand{\thelem}{S\arabic{lem}}
\renewcommand{\theprop}{S\arabic{prop}}
\renewcommand{\thecor}{S\arabic{cor}}
\renewcommand{\therem}{S\arabic{rem}}
\renewcommand{\thefigure}{S\arabic{figure}}
\renewcommand{\thetable}{S\arabic{table}}
\renewcommand{\theequation}{S\arabic{equation}}
\renewcommand{\thealgocf}{S\arabic{algocf}}

\begin{abstract}
	
	This supplement provides 
	(i) all technical details for the theorems of the conditional regularized halfspace depth given in \autoref{sec2} of the main paper (Section~\ref{sec_S1}),
	(ii) detalied simulation procedures for \autoref{sec4} of the main paper (Section~\ref{sec_S2}),  
	and (iii) all extra results regarding the numerical studies  (Section~\ref{sec_S3}).
	
	\vspace{0.1in}
	\noindent
	\textit{Keywords and phrases:} 
	conditional distributions;
	depth-based rank tests;
	functional rankings; 
	sparse functional data analysis.
	
\end{abstract}



\newpage


\section{Technical details} \label{sec_S1}

This section devotes to proving Theorems~\ref{thmNondege}--\ref{thmProperties} given in the main paper. 
Although some arguments overlap with those in \cite{YDL25RHD}, 
which we will abbreviate as [YDL25] in what follows,
the present sparse-data setting requires a more delicate treatment because of the conditional distributions appearing in the definition of CRHD. 
We therefore omit repetitive details and focus instead on the key distinctions from the dense-data setting considered therein.

\begin{proof}[Proof of \autoref{thmNondege}]
	We follow the proof of [YDL25,~Theorem~1]
	but need some modification due to conditional probability.
	Let $\mu = \eo[X]$ denote the mean of the random function $X$. 
	Since the Cauchy-Schwarz inequality implies
	\begin{align*}
		1 = \|v\| \leq \|\ga^{1/2}v\| \|\ga^{-1/2}v\| \leq \ld \|\ga^{1/2}v\|, \quad \forall v \in \vv_\ld,
	\end{align*}
	it holds that
	\begin{align*}
		|\langle x-\mu,v \rangle| \leq \|x-\mu\|\|v\| \leq \ld \|x-\mu\| \|\ga^{1/2}v\|, \quad \forall x \in \HH, \forall v \in \vv_\ld.
	\end{align*}
	Then, by the property of the cumulative distribution function, we have that
	\begin{align*}
		\inf_{v \in \vv_\ld} \pr ( \langle X, v \rangle \geq \langle X_0, v \rangle | X_0)
		& = \inf_{v \in \vv_\ld} \pr \left( {\langle X-\mu,v \rangle \over \|\ga^{1/2} v\|}  \geq {\langle X_0-\mu, v \rangle \over \|\ga^{1/2}v \|} \Big| X_0 \right)
		\\& \geq \inf_{v \in \vv_\ld} \pr \left( {\langle X-\mu,v \rangle \over \|\ga^{1/2} v\|}  \geq \ld \|X_0-\mu\| \Big| X_0 \right)
		\\& = \inf_{v \in \vv_\ld} (1-F_v(\ld\|X_0-\mu\|))
		\\& = 1- \sup_{v \in \vv_\ld}F_v(\ld\|X_0-\mu\|) >0,
	\end{align*}
	where $F_v$ denotes the cumulative distribution function of the standardized projection $\langle X-\mu,v \rangle / \|\ga^{1/2} v\|$ of $X$ onto $v$.
	By the condition in Equation~\eqref{condNondege} of the main paper, 
	it almost surely holds that
	\begin{align*}
		D_\ld(\bm{\xx}_0) 
		& = \inf_{v \in \vv_\ld} \pr ( \langle X, v \rangle \geq \langle X_0, v \rangle | \bm{\xx}_0)
		= \inf_{v \in \vv_\ld} \eo \left[ \pr(\langle X,v \rangle \geq \langle X_{0,\ap}, v \rangle | X_0) \Big| \bm{\xx}_0 \right]
		\\& \geq \eo \left[ \inf_{v \in \vv_\ld} \pr ( \langle X, v \rangle \geq \langle X_0, v \rangle | X_0) \Big| \bm{\xx}_0 \right]
		\\& \geq \eo \left[ 1- \sup_{v \in \vv_\ld}F_v(\ld\|X_0-\mu\|) \Big| \bm{\xx}_0 \right]  > 0.
	\end{align*}
\end{proof}

\begin{proof}[Proof of \autoref{thmProperties}]
	
	Let $H_{x,v} = \{ y \in \HH: \langle y, v \rangle \geq \langle x, v \rangle \}$ denotes the closed halfspace in $\HH$ attached to $x \in \HH$ with normal vector $v \in \HH$. 
	but upon modification to incorporate subtlety due to conditional probability.

	\begin{enumerate}[(a)]
		\item 
		
		%
		
		In page~11 in the supplement of [YDL25],
		the following two sets $E_1$ and $E_2$ are shown to be equal:
		\begin{align*}
			E_1 & \equiv \{ u \in \HH: \|u\|=1, \|\ga^{-1/2} u\| \leq \ld \}, \\
			E_2 & \equiv \{ A^\top v \in \HH: v \in \HH, \|v\|=1, \|\ga^{-1/2} (A^\top v)\| \leq \ld \},
		\end{align*}
		where $A^\top$ denotes the adjoint operator of $A$. 
		This implies that
		\begin{align*}
			D_\ld(\bm{\yy}_0; P_{AX+b})
			& = \inf \{ \pr( \langle AX+b, v \rangle \geq \langle Y_0, v \rangle |\bm{\yy}_0 ): v \in \HH, \|v\|=1, \|(A \ga A^\top)^{-1/2} v\| \leq \ld  \}
			\\& = \inf \{ \pr( \langle AX - Ax, v \rangle \geq 0|\bm{\xx}_0): v \in \HH, \|v\|=1, \|(A \ga A^\top)^{-1/2} v\| \leq \ld  \}
			\\& = \inf \{ \pr( \langle X - x, A^\top v \rangle \geq 0|\bm{\xx}_0): v \in \HH, \|v\|=1, \|\ga^{-1/2} (A^\top v)\| \leq \ld  \}
			\\& = \inf \{ \pr( \langle X - x, u \rangle \geq 0|\bm{\xx}_0): u \in \HH, \|u\|=1, \|\ga^{-1/2}u\| \leq \ld  \}
			\\& = D_\ld(\bm{\xx}_0; P_X)
		\end{align*}
		almost surely.
		
		\item 
		
		Let $\bm{\xx}_0$ be a sparse functional datum with true function $X_0$.
		If $X_0 = \theta$, then for all $v \in \HH$ with $\|v\|=1$, we have
		\begin{align*}
			\pr( \langle X , v \rangle \geq \langle X_0, v \rangle|X_0)
			= \pr( \langle X , v \rangle \geq \langle \theta, v \rangle)
			= 1/2,
		\end{align*}
		and hence, 
		\begin{align*}
			D_\ld(\bm{\xx}_0)
			& = \inf_{v \in \vv_\ld} \eo \left[ \pr( \langle X , v \rangle \geq \langle X_0, v \rangle|X_0) \Big|X_0 \right]
			= \inf_{v \in \vv_\ld} \eo[(1/2)|X_0]
			= 1/2
			\\& \leq D^{dense}_\ld(\theta).
		\end{align*}
		If $X_0 \neq \theta$, 
		then similarly to the argument in page~11 in the supplement of [YDL25],
		we can find $v \in \vv_\ld$ such that 
		$\theta \notin H_{X_0,v}$, $H_{\theta,-v} \subseteq H_{X_0,v}^c$,
		and
		\begin{align*}
			\pr(X \in H_{X_0,v}|\bm{\xx}_0) 
			\leq 1-\pr(X \in H_{\theta,-v}|\bm{\xx}_0) = 1-\pr(X \in H_{\theta,-v}) \leq 1/2,
		\end{align*}
		which implies that
		\begin{align*}
			D_\ld(\bm{\xx}_0)
			= \inf_{v \in \vv_\ld} \pr( \langle X , v \rangle \geq \langle X_0, v \rangle|\bm{\xx}_0) 
			\leq 1/2 
			\leq D^{dense}_\ld(\theta).
		\end{align*}
		Therefore, we conclude that $D_\ld(\bm{\xx}_0) \leq D^{dense}_\ld(\theta)$
		almost surely.

		\item 
		Without loss of generality, we may assume that $\ap \in (0,1)$.
		Let $\bm{\xx}_0$ be a sparse functional datum with true function $X_0$ and write $X_{0,\ap} = \theta + \ap(X_0 - \theta)$ and $\bm{X}_{0, \ap} = \bm{\theta} + \ap(\bm{\xx}_0 - \bm{\theta})$. 
		If $X_0 = \theta$, then $X_{0,\ap} = \theta = X_0$ and $\bm{X}_{0, \ap} = \bm{\theta} = \bm{\xx}_0$,
		meaning that $D_\ld(\bm{\xx}_0) = D_\ld(\bm{\xx}_{0,\ap})$.
		If $X_0 \neq \theta$, 
		then by Equation~(S8) in page~12 of the supplement of [YDL25], 
		for $v \in \vv_\ld$ with $\theta \notin H_{X_0,v}$, 
		we have $X_0 \in H_{\theta,v}$ and $H_{X_0,v} \subseteq H_{X_{0,\ap},v} \subseteq H_{\theta,v}$.
		This yields
		\begin{align*}
			\pr( \langle X, v \rangle \geq \langle X_{0,\ap}, v \rangle | \bm{\xx}_{0,\ap})
			= \pr( \langle X, v \rangle \geq \langle X_{0,\ap}, v \rangle | \bm{\xx}_0)
			\geq \pr( \langle X, v \rangle \geq \langle X_0, v \rangle | \bm{\xx}_0).
		\end{align*}
		That is, if $X_0 \neq \theta$, we have $D_\ld(\bm{\xx}_0) \leq D_\ld(\bm{\xx}_{0,\ap}).$
		Combining these two, we conclude that $D_\ld(\bm{\xx}_0) \leq D_\ld(\bm{\xx}_{0,\ap})$ almost surely.
		
		\item 
		
		%
		
		Following the proof of [YDL25,~Theorem~3(d)],
		we may assume that $\W \cap \vv_\ld \neq \emptyset$ with $\dim \W = L$ and define an isometry map $\ttt:\W \to \HH_L \equiv \mathrm{span}(\{\phi_l\}_{l=1}^L) $ as $\ttt \psi_l \equiv \phi_l$ for each $l=1,\dots,L$.
		Let $Y = \ttt \Pi_\W X$ and $Y_{0n} = \ttt \Pi_\W X_{0n}$		
		so that $\|Y_{0n}\| = \|\ttt \Pi_\W X_{0n}\| = \| \Pi_\W X_{0n}\| \to \infty$ as $n\to\infty$ almost surely. 
		By using Lemma~S3 in the supplement of [YDL25] 
		and following the same argument as the proof of Theorem~3(d) therein,
		we can find that $U_n$ with $\|U_n\|=1$ such that $\langle Y_{0n}, U_n \rangle \to \infty$ as $n\to\infty$ almost surely
		and that the following inequality holds:
		\begin{align*}
			D_\ld(\xx_{0n}) \leq \pr(\|Y\| \geq \langle Y_{0n}, U_n \rangle| \xx_{0n} ).
		\end{align*}
		The convergence to zero of the upper bound in the preceding display follows from the tightness of $Y$ \cite[cf.][Theorem~1.3]{bill99} and bounded convergence theorem for conditional expectation
		\cite[e.g.,][Theorem~12.2.3]{AL06}.
	\end{enumerate}
\end{proof}


\section{Detailed simulation procedures} \label{sec_S2}

In this section, upon introducing common setups in \autoref{ssec_S2_1}, we provide detailed procedures of each simulation study in Sections~\ref{ssec_S2_2}--\ref{ssec_S2_4}.

\subsection{Setups and notations} \label{ssec_S2_1}

While the number $K$ of functional principal components (FPCs) used in estimation is suppressed in the main paper,
we explicitly write $K$ for more accurate description of implementation.
Also, in the numerical studies, since we control the quantile level $u$ that determines the regularization $\ld$, instead of directly manipulating $\ld$ itself, we denote the RHDs with $u$ rather than $\ld$. 

The two-stage approaches under our comparison employed 
the principal analysis by conditional expectation (PACE) technique \citep{YMW05a, YMW05b, LH10, ZW16,DMT18,GDM21}
to extract the FPC scores and to obtain the smoothed curves. 
These two reconstructed data were computed using the \texttt{fdapace} package \citep{fdapace}.

\newcommand{\Xscr}{\mathscr{X}}

Suppose that the sparse functional data $\Xscr_n \equiv \{\bm{\xx}_i = (\bm{T}_i, \bm{X}_i)\}_{i=1}^n$ are observed following \eqref{eqObsData}, 
with the vectors $\bm{T}_i = [T_{i1}, \dots, T_{in_i}]^\top$ and $\bm{X}_i = [\tilde{X}_{i1}, \dots, \tilde{X}_{in_i}]^\top$ of measurement times and the observe data,
where the underlying curves $\X_n \equiv \{X_i\}_{i=1}^n$ are independently and identically generated by the $K^*$-truncated Karhunen--Lo\`{e}ve expansion in  \eqref{eqKLtrunc} with $K^* = 15$. 
We describe the two-stage approaches under comparison in more details.
\begin{enumerate}[(a)]
	\item (TwoStageTHD, applying multivariate depth to the first FPC scores)
	Let $\bm{\xi}_{iK} \equiv (\xi_{i1}, \dots, \xi_{iK})^\top$ denote the vector of the first $K$ FPC scores of the $i$-th sparse observation $\bm{\xx}_i$ computed by the PACE. 
	We can construct the finite-dimensional halfspace depth \citep{tukey75} based on $\Xi_{nK} \equiv \{\bm{\xi}_{iK}\}_{i=1}^n$, denoted as $D_K(\cdot; \Xi_{nK})$.
	This is approximated by random projections \citep{CN08} and implemented in \texttt{ddalpha} package, denoted by $\tilde{D}_{K,L}(\cdot; \Xi_{nK})$, 
	although the approximations may not be sufficient for large $K$, e.g., $K \geq 4$. 
	For a sparse observation $\bm{\xx}_0$, denoting its first $K$ FPC scores as $\bm{\xi}_{0K}$, 
	the (approximated sample) truncated halfspace depth at $\bm{\xx}_0$ is defined as $\tilde{D}_{\mathrm{Tukey}, \infty, K, L}(\bm{\xx}_0; \Xscr_n) \equiv \tilde{D}_{K,L}(\bm{\xi}_{0K}; \Xi_{nK})$. 
	
	\item (TwoStageRHD, applying dense functional depth to the predicted curves)
	Let $\hat{\X}_n \equiv \{\hat{X}_i\}_{i=1}^n$ denote the set of the predicted curves using the PACE, where the number of FPCs for the prediction is chosen as the possible maximum by putting  \texttt{optns = list(maxK = min(n,49))} in \texttt{fdapace::FPCA()} function.
	We can construct the original RHD $x \mapsto \tilde{D}_{u, K L}^{\mathrm{dense}}(x; \hat{\X}_n)$  based on $\hat{\X}_n$.
	For a sparse observation $\bm{\xx}_0$, denoting its predicted curve as $\hat{X}_0$, 
	the (approximated sample) smoothed RHD at $\bm{\xx}_0$ is defined as $\tilde{D}_{\mathrm{pred}, u, K, L}(\bm{\xx}_0; \Xscr_n) \equiv \tilde{D}_{u, K L}^{\mathrm{dense}}(\hat{X}_0; \hat{\X}_n)$. 
\end{enumerate}
The averaged and plug-in CRHDs in the main paper are respectively denoted as $\tilde{D}_{\mathrm{ave}, u, K, L}(\cdot; \Xscr_n)$ and $\tilde{D}_{\mathrm{plug}, u, K, L}(\cdot; \Xscr_n)$; 
the PACE were used again to estimate the mean and covariance functions of $X$.
The final set of sparse functional depths under comparison is 
\begin{align} \label{eq_depth_final}
	\tilde{D}_{\mathrm{ave}, u, K, L}(\cdot; \Xscr_n), \quad
	\tilde{D}_{\mathrm{plug}, u, K, L}(\cdot; \Xscr_n), \quad
	\tilde{D}_{\mathrm{pred}, u, K, L}(\cdot; \Xscr_n), \quad
	\tilde{D}_{\mathrm{Tukey}, \infty, K, L}(\cdot; \Xscr_n).
\end{align}
As common tuning parameters are $K \in \{2,4,6,8,10\}$ and $L = 1000K$, we write $\tilde{D}_{q, u, K, L}(\cdot; \Xscr_n)$ to denote each sparse functional depth, where $q$ varies in $Q \equiv \{\mathrm{ave}, \mathrm{plug}, \mathrm{pred}, \mathrm{Tukey} \}$. When $q \neq \mathrm{Tukey}$, the same quantile level $u \in \{ 0.4, 0.6, 0.8, 0.95 \}$ will be equally applied to each RHDs, while $q = \mathrm{Tukey}$ automatically indicates $u = \infty$. 
We apply the following procedure to each halfspace depth with the same choices of $u, K,L$ as described in the main text.

\subsection{Ranking recovery} \label{ssec_S2_2}

The simulation procedure of ranking recovery resembles the one studied by \citet[Section~4]{SL21}.
In a nutshell, we compared correlations between rankings from the (dense) RHD values of the underlying true functions and each of the RHDs for sparse observations. 
For each $m=1,\dots,M$, do the following. 
\begin{enumerate}
	\item 
	Construct the approximated sample RHD functions $x \mapsto \tilde{D}_{u, K, L}^{\mathrm{dense}}(x; \X_n)$ based on the true curves $\X_n$, as described in [YDL25,~Algorithm~1],
	and compute the RHD values for each $X_i$, i.e., $\{d_{\mathrm{true}, i} \equiv \tilde{D}_{u, K, L}^{\mathrm{dense}}(X_i; \X_n)\}_{i=1}^n$ 
	
	\item
	
	Construct the first three sparse functional depths in \eqref{eq_depth_final}
	and compute the RHD values $\{ d_{q,i} \equiv \tilde{D}_{q, u, K, L}(\bm{\xx}_i; \Xscr_n) \}_{i=1}^n$ for each $q \in \{\mathrm{ave}, \mathrm{plug}, \mathrm{pred}\}$.

	\item 
	For each $q \in \{\mathrm{true}, \mathrm{ave}, \mathrm{plug}, \mathrm{pred}\}$, compute the rankings $\{r_{q,i}\}_{i=1}^n$ of depth values,
	where $r_{q,i}$ denotes the ranking of $d_{q,i}$ and a lower ranking means a lower depth value;
	in case of ties, we take the minimum. 
	
	\item 
	Compute the empirical correlation between $\{r_{\mathrm{true}, i}\}_{i=1}^n$ and $\{r_{q,i}\}_{i=1}^n$ (i.e., the Spearman correlation between depth values $\{d_{\mathrm{true}, i}\}_{i=1}^n$ and $\{d_{q,i}\}_{i=1}^n$) for each $q \in  \{\mathrm{pred}, \mathrm{ave}, \mathrm{plug}\}$ by using \texttt{stats::cor()} in R.
	The resulting correlation is respectively written by $C_{q, m}$ for each $q \in  \{\mathrm{ave}, \mathrm{plug}, \mathrm{pred}\}$.

\end{enumerate}
For each $q \in  \{\mathrm{ave}, \mathrm{plug}, \mathrm{pred}\}$, we compute the average $\cc_q \equiv M^{-1} \sum_{m=1}^M C_{q, m}$ of these correlations $\{C_{q, m}\}_{m=1}^M$ as a measure of ranking recovery.
In \autoref{figSRHDrankcorr} of the  main paper, ACRHD, PCRHD, and TwoStageRHD respectively represent the values of $\cc_{\mathrm{ave}}, \cc_{\mathrm{plug}}, \cc_{\mathrm{pred}}$.

\subsection{Rank tests} \label{ssec_S2_4}

The depth-based rank test we conducted follows a modified Kruskal--Wallis test by \cite{CS12}, which is improved from the original one by \cite{LS93, ZH06}.
We write $\chi^2_{1-\ap}(1)$ to denote the $1-\ap$ quantile of the chi-square distribution with degree of freedom 1.
For each $m=1,\dots,M$, perform the following. 
\begin{enumerate}
	\item 
	Generate the two samples $\Xscr_{0n} \equiv \{\xx_{0i}\}_{i=1}^{n_0}$ and $\Xscr_{1n} \equiv \{\xx_{1i}\}_{i=1}^{n_1}$ of sparse observations, as described in \autoref{ssec_4_4} of the main paper.

	\item 
	For each $q \in Q$, construct the depth $\tilde{D}_{q, u, K, L|g} \equiv \tilde{D}_{q, u, K, L}(\cdot; \Xscr_{gn})$ based on the sample $\Xscr_{gn} \equiv \{\xx_{gi}\}_{i=1}^{n_g}$, where $g \in \{0,1\}$ denotes the group,
	and compute the depth values $d_{q,g,i|0} =\tilde{D}_{q, u, K, L|0}(\bm{\xx}_{gi}) $ and $d_{q,g,i|1} =\tilde{D}_{q, u, K, L|1}(\bm{\xx}_{gi})$ at all observations using each depth.
	
	\item 
	Compute the rankings $r_{q,g,i|0}$ of $d_{q,g,i|0}$ in $\bigcup_{g=0}^1 \{d_{q,g,i|0}\}_{i=1}^{n_0}$ based on the first group $\Xscr_{0n}$, and the rankings $r_{q,g,i|1}$ of $d_{q,g,i|1}$ in $\bigcup_{g=0}^1 \{d_{q,g,i|1}\}_{i=1}^{n_1}$ based on the second group $\Xscr_{1n}$; a lower ranking means a lower depth.
	
	\item 
	Compute the classical Kruskal--Wallis test statsitics \citep{Krusk52, KW52} using \texttt{stats::kruskal.test()} in R, based on the depth rankings $\bigcup_{g=0}^1 \{r_{q,g,i|0}\}_{i=1}^{n_0}$ and $\bigcup_{g=0}^1 \{r_{q,g,i|1}\}_{i=1}^{n_1}$, say $H_{q,0}$ and $H_{q,1}$. 
	
	\item 
	The final statistic is computed by their average as $H_q \equiv (H_{q,0} + H_{q,1}) / 2$,
	and check if $H_q$ is in the rejection region $(\chi^2_{1-\ap}(1), \infty)$ by computing the indicator	$I_{q,m}\equiv \I(H_q > \chi^2_{1-\ap}(1))$. 
	
\end{enumerate}
For each $q \in Q$, we compute the average of the indicators $\{I_{q,m}\}_{m=1}^\infty$ to approximate the rejection rate of the corresponding rank test: 
\begin{align*}
	\widehat{RR}_q \equiv M^{-1} \sum_{m=1}^M  I_{q,m}. 
\end{align*}


\section{Additional simulation results} \label{sec_S3}

All extra figures regarding simulation studies are collected in this section.

\begin{sidewaysfigure}
	\centering
	\includegraphics[width=0.95\linewidth]{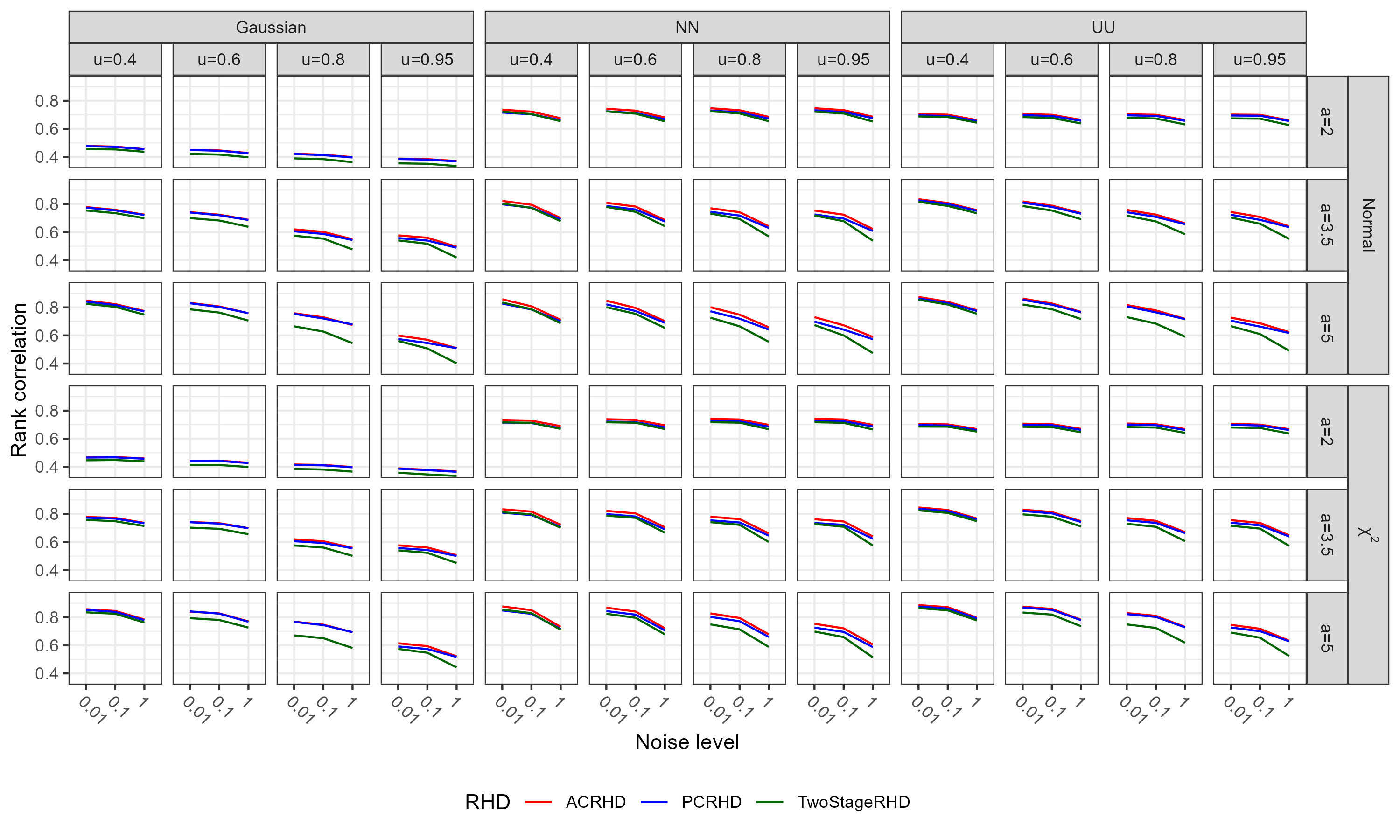}
	\caption{
			Empirical rank correlations of ACRHD, PCRHD, and TwoStageRHD values at sparse observations with the original RHD values at the true underlying functions
			when the sample size is $n=50$ and the RHDs are constructed with $K=4$.
		}
	\label{fig_recover1_n50}
\end{sidewaysfigure}

\begin{sidewaysfigure}
	\centering
	\includegraphics[width=0.95\linewidth]{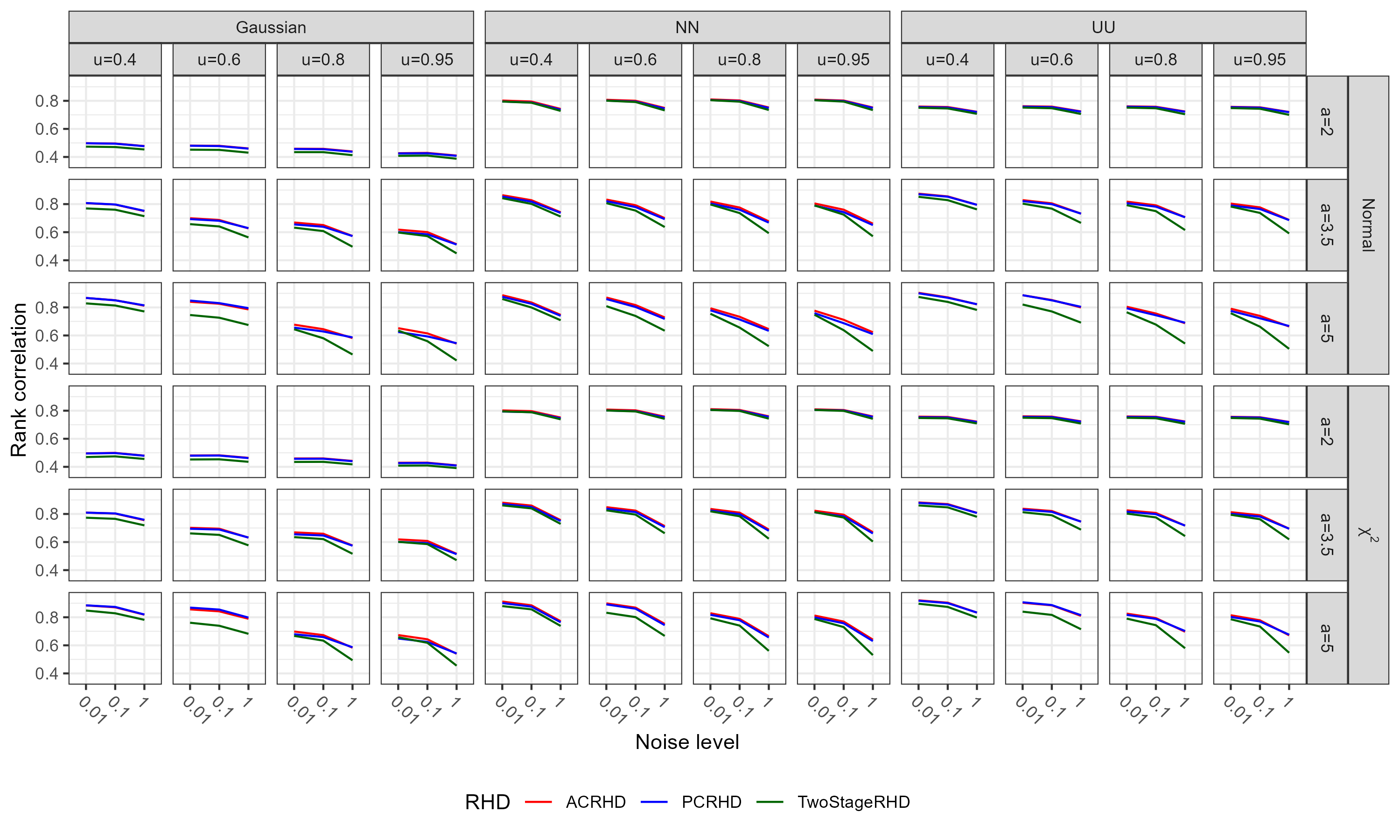}
	\caption{
			Empirical rank correlations of ACRHD, PCRHD, and TwoStageRHD values at sparse observations with the original RHD values at the true underlying functions
			when the sample size is $n=200$ and the RHDs are constructed with $K=6$.
		}
	\label{fig_recover2_n200}
\end{sidewaysfigure}

\begin{sidewaysfigure}
	\centering
	\includegraphics[width=0.95\linewidth]{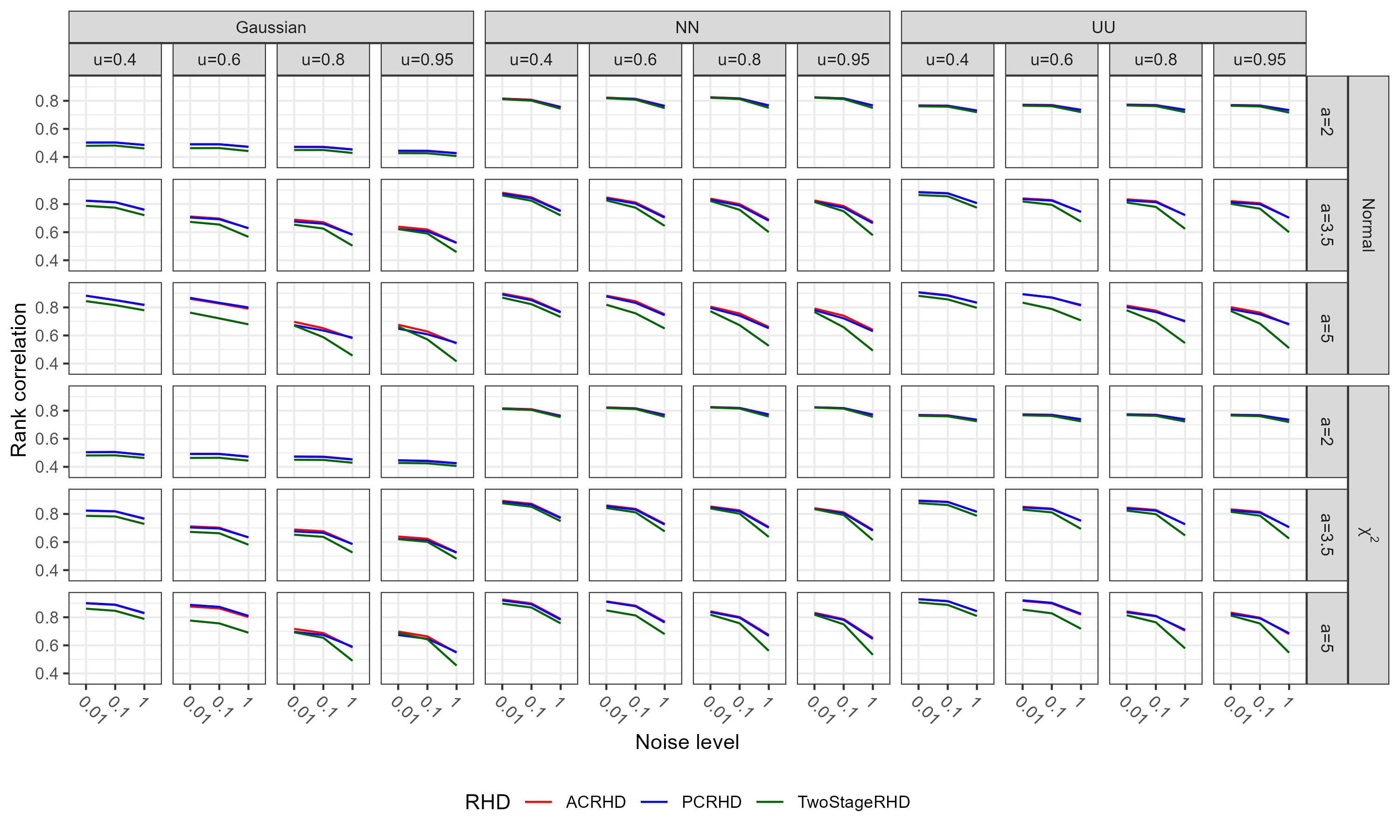}
	\caption{
			Empirical rank correlations of ACRHD, PCRHD, and TwoStageRHD values at sparse observations with the original RHD values at the true underlying functions
			when the sample size is $n=400$ and the RHDs are constructed with $K=6$.
		}
	\label{fig_recover3_n400}
\end{sidewaysfigure}


%
%

\begin{sidewaysfigure}
	\centering
	\includegraphics[width=0.85\linewidth]{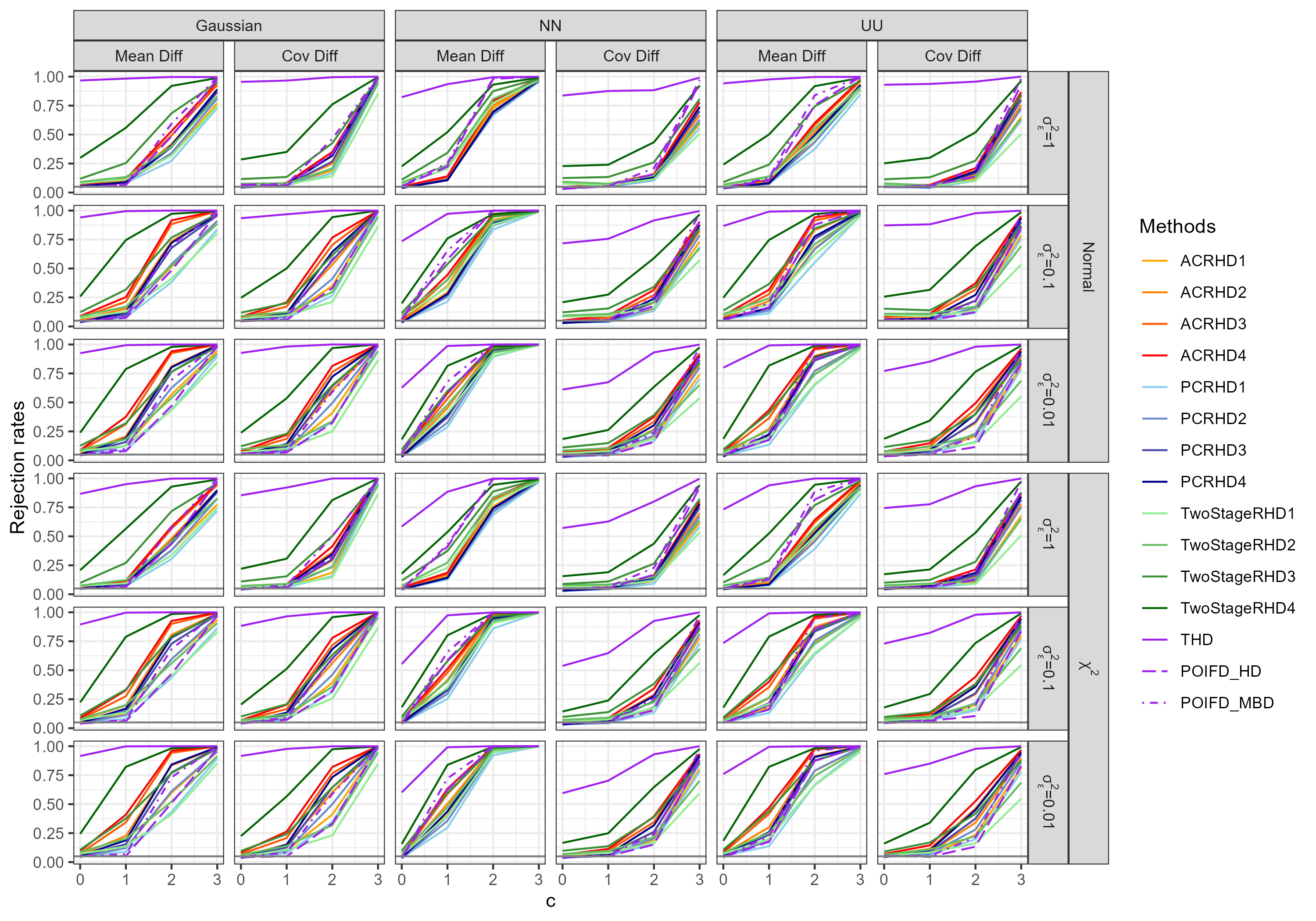}
	\caption{
			Empirical rejection rates  of the KW tests based on both empirical and plug-in SRHDs  as the degree $c \in \{0, 1, 2, 3\}$ of the alternative increases; the gray horizontal line represents the nominal size 5\%.
			The total sample size is $n=100$ and all depths are computed by taking $K=6$. 
			From RHD1 to RHD4, the quantile levels $u \in \{0.4, 0.6, 0.8, 0.95\}$ are used for the RHDs.
		}
	\label{fig_rank1_n100}
\end{sidewaysfigure}

\begin{sidewaysfigure}
	\centering
	\includegraphics[width=0.85\linewidth]{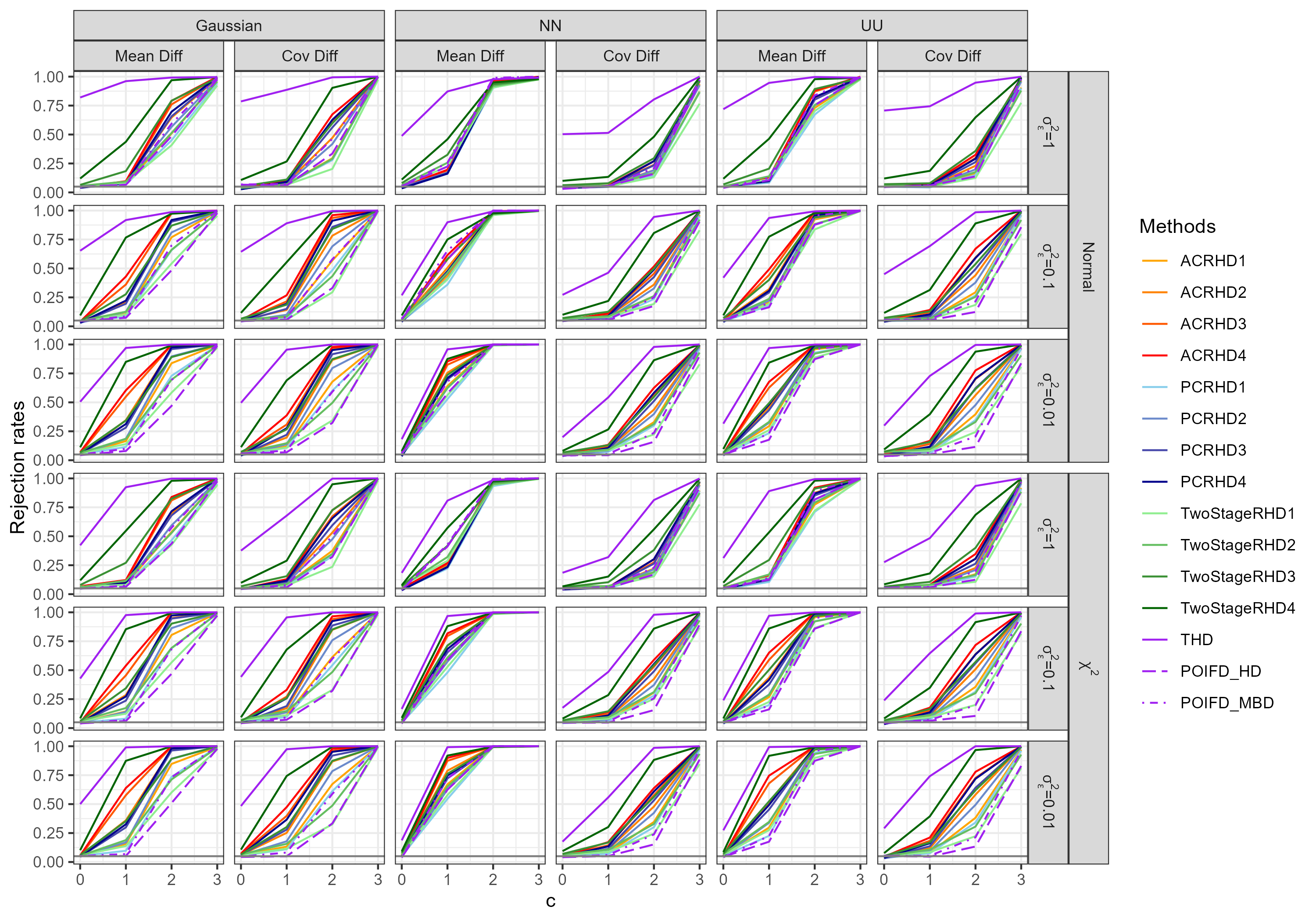}
	\caption{
			Empirical rejection rates  of the KW tests based on both empirical and plug-in SRHDs  as the degree $c \in \{0, 1, 2, 3\}$ of the alternative increases; the gray horizontal line represents the nominal size 5\%.
			The total sample size is $n=200$ and all depths are computed by taking $K=6$. 
			From RHD1 to RHD4, the quantile levels $u \in \{0.4, 0.6, 0.8, 0.95\}$ are used for the RHDs.
		}
	\label{fig_rank2_n200}
\end{sidewaysfigure}

\end{document}